\documentclass[11pt,a4paper,reqno]{amsart}
\usepackage{amsthm,amsmath,amsfonts,amssymb,amsxtra,appendix,bookmark,euscript,delarray,dsfont}
\usepackage{enumerate}
\usepackage{color,graphicx}
\usepackage{tikz}
\usepackage[latin1]{inputenc}

\theoremstyle{plain}
\newtheorem{theorem}{Theorem}
\newtheorem{lemma}[theorem]{Lemma}
\newtheorem*{lemma*}{Lemma}

\newtheorem{proposition}[theorem]{Proposition}

\newtheorem*{conjecture*}{Conjecture}

\theoremstyle{definition}

\theoremstyle{remark}
\newtheorem{remark}{Remark}

\DeclareMathOperator{\Tr}{Tr}

\DeclareMathOperator{\loc}{loc}
\DeclareMathOperator{\Span}{span}

\def\leqslant{\le}
\def\bq{\begin{eqnarray}}
\def\eq{\end{eqnarray}}
\def\bqq{\begin{eqnarray*}}
\def\eqq{\end{eqnarray*}}
\def\nn{\nonumber}

\def\eps{\varepsilon}

\newcommand{\norm}[1]{\left\lVert #1 \right\rVert}
\renewcommand{\phi}{\varphi}
\newcommand\1{{\ensuremath {\mathds 1} }}

\def\GN{\mathrm{GN}}
\def\LT{\mathrm{LT}}
\def\BLT{\mathrm{BLT}}

\def\sym{\mathrm{sym}}

\def\C{\mathbb{C}}
\def\N{\mathbb{N}}
\def\R{\mathbb{R}}

\def\cD{\mathcal{D}}
\def\cF{\mathcal{F}}
\def\cH{\mathcal{H}}

\def\V {\mathcal{V}}

\def\cN{\mathcal{N}}

\newcommand{\bDelta}{{\mbox{$\triangle$}\hspace{-8.0pt}\scalebox{0.8}{$\triangle$}}}
\allowdisplaybreaks

\title[Fractional Hardy-Lieb-Thirring and related inequalities]{Fractional Hardy-Lieb-Thirring and related inequalities for interacting systems}

\author[D. Lundholm]{Douglas Lundholm}
\address{KTH Royal Institute of Technology, Sweden} 
\email{dogge@math.kth.se,}

\author[P.T. Nam]{Phan Th\`anh Nam}
\address{Institute of Science and Technology Austria, Austria} 
\email{pnam@ist.ac.at}

\author[F. Portmann]{ Fabian Portmann}
\address{University of Copenhagen, Denmark} 
\email{fabian@math.ku.dk}

\begin{document}
\date{August, 2015}

\begin{abstract}
	We prove analogues of the Lieb-Thirring and Hardy-Lieb-Thirring 
	inequalities for many-body quantum systems with fractional kinetic 
	operators and homogeneous interaction potentials, where no anti-symmetry 
	on the wave functions is assumed. 
	These many-body inequalities imply interesting one-body interpolation 
	inequalities, and we show that the corresponding one- and 
	many-body inequalities are actually equivalent in certain cases.
\end{abstract}

\maketitle

\section{Introduction}

The uncertainty principle and the exclusion principle are two of the most 
important concepts of quantum mechanics. 
In 1975, Lieb and Thirring \cite{LieThi-75,LieThi-76} gave an elegant 
combination of these principles in a semi-classical lower bound on the 
kinetic energy of fermionic systems. 
They showed that there exists a constant $C_{\rm LT}>0$ depending only on 
the dimension $d \ge 1$ such that the inequality
\begin{align} \label{eq:LT-fermion}
	\left\langle \Psi, \sum_{i=1}^N -\Delta_{i} \Psi \right\rangle 
	\ge C_{\LT} \int_{\R^d} \rho_\Psi(x)^{1+2/d} \,dx
\end{align}
holds true for every function $\Psi\in H^1((\R^d)^N)$ and for all 
$N \in \N$, provided that $\Psi$ is normalized and anti-symmetric, 
namely $\|\Psi\|_{L^2(\R^{dN})}=1$ and 
\begin{align} \label{eq:anti-symmetry}
	\Psi(x_1,\dots,x_i,\dots,x_j,\dots,x_N)= - \Psi(x_1,\dots,x_j,\dots,x_i,\dots,x_N),\quad \forall i\ne j.
\end{align}
The left hand side of \eqref{eq:LT-fermion}
is the expectation value of the kinetic energy
operator for $N$ particles, and
for every $N$-body wave function $\Psi \in L^2((\R^d)^N)$, 
its one-body density is defined by 
$$
	\rho_\Psi(x) := \sum_{j=1}^N \int_{\R^{d(N-1)}} |\Psi(x_1,\dots,x_{j-1},x,x_{j+1},\dots,x_N)|^2 \prod\limits_{i\ne j} dx_i.
$$
Note that $\int_{Q}\rho_{\Psi}$ can be interpreted as the expected number of 
particles to be found on a subset $Q \subset \R^d$ in the probability 
distribution given by $|\Psi|^2$. In particular, $\int_{\R^d}\rho_{\Psi} = N$.

The Lieb-Thirring inequality can be seen as a many-body generalization of the Gagliardo-Nirenberg inequality
\begin{align} \label{eq:GN}
	\left( \int_{\R^d} |\nabla u(x)|^{2} dx \right) 
	\left( \int_{\R^d} |u(x)|^2 dx \right)^{2/d} 
	\ge C_{\GN} \int_{\R^d} |u(x)|^{2 (1+2/d)} dx,
\end{align}
for $u\in H^1(\R^d)$. Note that for $d\geq3$, the Gagliardo-Nirenberg inequality \eqref{eq:GN} is a consequence of 
Sobolev's inequality
\begin{align}\label{eq:Sobolev}
	\| \nabla u \|_{L^2(\R^d)} \ge C_{\rm S} \|u \|_{L^{2d/(d-2)}(\R^d)}
\end{align}
and the H\"older interpolation inequality for $L^p$-spaces. 
Moreover, 
Sobolev's inequality can actually be obtained from Hardy's inequality
\begin{align}\label{eq:Hardy}
	\| \nabla u \|_{L^2(\R^d)}^2 \ge \frac{(d-2)^2}{4} \int_{\R^d} \frac{|u(x)|^2}{|x|^2}\,dx,
	\quad d > 2,
\end{align}
by a symmetric-decreasing rearrangement argument
(see, e.g., \cite[Sec.~4]{FraSei-08}).
 
All of the 
inequalities \eqref{eq:GN}-\eqref{eq:Sobolev}-\eqref{eq:Hardy} are 
quantitative formulations of the uncertainty principle. 
On the other hand, the anti-symmetry
\eqref{eq:anti-symmetry}, which is crucial for the 
Lieb-Thirring inequality \eqref{eq:LT-fermion} to hold, 
corresponds to {Pauli's exclusion principle} for fermions. 
In fact, inequality \eqref{eq:LT-fermion} fails to apply to the 
product wave function 
$$\Psi(x_1,x_2,\dots,x_N)=u(x_1)u(x_2)\cdots u(x_N)=:u^{\otimes N}(x_1,x_2,\dots,x_N),$$
which is a typical state of 
bosons\footnote{In general, bosonic wave functions satisfy 
\eqref{eq:anti-symmetry} with a plus instead of a minus sign.}. 
In this case $\rho_{u^{\otimes N}}(x) = N |u(x)|^2$ and we only have
the weaker inequality
\begin{align}\label{eq:LT-fermion-uN}
	\left\langle u^{\otimes N}, \Big(\sum_{i=1}^N -\Delta_{i} \Big) u^{\otimes N} \right\rangle 
	\ge C N^{-2/d} \int_{\R^d} \rho_{u^{\otimes N}}(x)^{1+2/d}dx,
\end{align}
which is, however, {\em equivalent} to the Gagliardo-Nirenberg inequality 
\eqref{eq:GN}.

The discovery of Lieb and Thirring goes back to the stability of matter problem 
(see \cite{LieSei-10} for a pedagogical introduction to this subject). 
It is often straightforward to derive the finiteness of the ground state energy 
of quantum systems 
from a formulation of the uncertainty principle such as 
\eqref{eq:GN}, \eqref{eq:Sobolev} or \eqref{eq:Hardy}.
However, the fact that 
the energy does not diverge faster than proportionally to the number of particles 
--- that is, stability in a thermodynamic sense --- 
is much more subtle and for this the exclusion principle is crucial. 
It was Dyson and Lenard \cite{DysLen-67,DysLen-68} who first proved 
thermodynamic stability for fermionic Coulomb systems, 
and their proof is based on 
a \emph{local} formulation of the exclusion principle,
which is a relatively weak consequence of \eqref{eq:anti-symmetry}. 
Later Lieb and Thirring \cite{LieThi-75} gave a much shorter proof of the 
stability of matter using their more powerful inequality \eqref{eq:LT-fermion}.

Recently, Lundholm and Solovej \cite{LunSol-13} realized that the local 
exclusion principle in the original work of 
Dyson and Lenard \cite{DysLen-67,DysLen-68}, 
when combined with local formulations of the uncertainty principle, 
actually implies 
the Lieb-Thirring inequality \eqref{eq:LT-fermion}. 
From this point of view, they derived Lieb-Thirring inequalities for anyons, 
two-dimensional particles which do not satisfy the full anti-symmetry \eqref{eq:anti-symmetry}
but still fulfill a fractional exclusion. 
The same approach was also employed to prove Lieb-Thirring inequalities for 
fractional statistics particles in one dimension by the same authors 
\cite{LunSol-14}, as well as for fermions with certain point interactions 
by Frank and Seiringer \cite{FraSei-12}. 

Following the spirit in \cite{LunSol-13}, Lundholm, Portmann and Solovej \cite{LunPorSol-14} 
found that Lieb-Thirring type inequalities still hold true for particles 
without any symmetry assumptions 
--- and therefore in particular for bosons ---
provided that the exclusion principle is replaced by a sufficiently strong 
repulsive interaction between particles. 
For example, they proved that there exists a constant $C>0$ depending only 
on the dimension $d\ge 1$ such that for every normalized function 
$\Psi\in H^1((\R^d)^N)$ and all $N \in \N$,
\begin{align} \label{eq:LT-boson}
	\left\langle \Psi, \left( \sum_{i=1}^N -\Delta_{i} + \sum_{1\le i<j \le N} \frac{1}{|x_i-x_j|^{2}} \right) \Psi \right\rangle 
	\ge C \int_{\R^3} \rho_{\Psi}(x)^{1+2/d} \,dx.
\end{align}

\noindent The appearance of the inverse-square interaction in \eqref{eq:LT-boson} 
is natural as it makes  all terms in the inequality scale 
in the same way.
\medskip

The aims of our paper are threefold. 
\smallskip

$\bullet$ We generalize the Lieb-Thirring inequality \eqref{eq:LT-boson} to the 
fractional kinetic operator $(-\Delta)^s$ for an arbitrary power $s>0$, 
with matching interaction $|x-y|^{-2s}$. 
The non-local property of $(-\Delta)^s$ for non-integer $s$
makes the inequality more involved. 
Nevertheless, the fermionic analogue of this inequality 
(without the interaction term) has been known for a long time in the context
of relativistic stability \cite{Daubechies-83}. 
For the interacting bosonic version
we will follow the strategy of \cite{LunPorSol-14}, using local uncertainty 
and exclusion, but we also develop several new tools.
In particular, we will introduce a new covering lemma 
which provides an elegant way to combine the local uncertainty and exclusion 
into a single bound.
\smallskip

$\bullet$ We prove a stronger version of the Lieb-Thirring inequality 
\eqref{eq:LT-boson} with the kinetic operator replaced by
$(-\Delta)^s-\mathcal{C}_{d,s}|x|^{-2s}$ and with
the interaction $|x-y|^{-2s}$, 
for all $0<s<d/2$. Here $\mathcal{C}_{d,s}$ is the optimal constant in the 
Hardy inequality \cite{Herbst-77}
$$(-\Delta)^s-\mathcal{C}_{d,s}|x|^{-2s} \ge 0.$$
Our result can be seen as a bosonic analogue to the Hardy-Lieb-Thirring 
inequality for fermions found by Ekholm, Frank, Lieb and Seiringer 
\cite{EkhFra-06,FraLieSei-07,Frank-09}. 

\smallskip

$\bullet$ Just as the Lieb-Thirring inequality \eqref{eq:LT-fermion}
implies the one-body interpolation inequality \eqref{eq:GN}, the same 
will be shown to be true for these generalized many-body inequalities.
For instance, our bosonic Hardy-Lieb-Thirring inequality implies the
one-body interpolation inequality
\begin{multline*}
	\left\langle u, \Big((-\Delta)^s - \mathcal{C}_{d,s} |x|^{-2s} \Big)  u \right\rangle^{1-2s/d} 
	\left( \iint_{\R^d \times \R^d} \frac{|u(x)|^2|u(y)|^2}{|x-y|^{2s}}\,dxdy \right)^{2s/d} \\
	\ge  C \int_{\R^d} |u(x)|^{2(1+2s/d)}\,dx,
\end{multline*}
for $u \in H^s(\R^d)$ and $0<s<d/2$.
Moreover,
we prove the {\em equivalence} between the (bosonic) 
Lieb-Thirring/Hardy-Lieb-Thirring inequalities and the corresponding 
one-body interpolation inequalities when $0<s\le 1$.
Since one-body interpolation inequalities have been studied actively for a 
long time, we believe that this equivalence could inspire many new 
directions to the many-body theory.  

\smallskip

In the next section our results will be presented in detail and an
outline of the rest of the paper given.

\smallskip

\noindent\textbf{Acknowledgment.}
We thank Jan Philip Solovej, Robert Seiringer and Vladimir Maz'ya for helpful discussions, 
as well as Rupert Frank and the anonymous referee for useful comments.
Part of this work has been carried out during a visit at the Institut Mittag-Leffler (Stockholm).
D.L. acknowledges financial support by the grant KAW 2010.0063 from the Knut and Alice Wallenberg Foundation
and the Swedish Research Council grant no. 2013-4734.
P.T.N. is supported by 
the People Programme (Marie Curie Actions) of the European Union's Seventh Framework 
Programme (FP7/2007-2013) under REA grant agreement no. 291734. 
F.P. acknowledges support from the ERC project no. 321029 ``The mathematics of the structure of matter".

\section{Main results}
\subsection{Fractional Lieb-Thirring inequality} \label{ssec:FLT}
Our first aim of the present paper
is to generalize \eqref{eq:LT-boson} to the fractional kinetic operator 
$(-\Delta)^s$ for an arbitrary power $s>0$, 
and with a matching interaction $|x-y|^{-2s}$.
The operator $(-\Delta)^s$ is defined as the multiplication operator $|p|^{2s}$ 
in Fourier space, namely
$$
	\left[(-\Delta)^s f\right]^\wedge\!(p)= |p|^{2s} \widehat{f}(p), 
	\quad \widehat{f}(p):= \frac{1}{(2\pi)^{d/2}}\int_{\R^d}f(x)e^{-ip\cdot x}\,dx.
$$
The associated space $H^s(\R^d)$ is a Hilbert space with norm
$$
	\|u\|_{H^s(\R^d)}^2 := \|u\|_{L^2(\R^d)}^2 + \|u\|_{\dot H^s(\R^d)}^2,
	\qquad
	\|u\|_{\dot H^s(\R^d)}^2 := \langle u, (-\Delta)^s u\rangle,
$$
and the addition of a positive interaction potential is to be understood 
as the sum of non-negative forms.

Our first result is the following
\begin{theorem}[Fractional Lieb-Thirring inequality]\label{thm:LT_frac}
	For all $d\ge 1$ and $s>0$, there exists a constant
	$C>0$ depending only on $d$ and $s$ such that for all $N\in \N$ 
	and for every $L^2$-normalized function 
	$\Psi\in H^s(\R^{dN})$,
	\begin{align}\label{eq:LT_frac} 
		\left\langle \Psi, \left( \sum_{i=1}^N (-\Delta_{i})^s 
		+ \sum_{1\le i<j \le N} \frac{1}{|x_i-x_j|^{2s}} \right) \Psi \right\rangle 
		\ge C \int_{\R^d} \rho_{\Psi}(x)^{1+2s/d}\,dx.
	\end{align}
\end{theorem} 
Since our result 
holds without restrictions on the symmetry of the wave function, 
and therefore in particular also for bosons,
we consider it as a 
bosonic analogue to the fermionic inequality\footnote{Throughout 
	our paper, $C$ denotes a generic positive constant. 
	Two $C$'s in different places may refer to two different constants.}
\begin{align}\label{eq:LT_frac_fermions} 
	\left\langle \Psi, \sum_{i=1}^N (-\Delta_{i})^s \Psi \right\rangle 
	\ge C \int_{\R^d} \rho_{\Psi}(x)^{1+2s/d} \,dx,
\end{align}
which holds for wave functions $\Psi$ satisfying the anti-symmetry \eqref{eq:anti-symmetry}, 
where the constant $C>0$ is independent of $N$ and $\Psi$.
The original motivation for such a fermionic fractional Lieb-Thirring inequality 
has been its usefulness in the context of 
stability of relativistic matter
(see \cite{Daubechies-83} and the recent review \cite{LieSei-10}). 
Our inequality \eqref{eq:LT_frac} for $s=1/2$ and $d=3$ is relevant to 
the physical situation of relativistic particles 
(which could be identical bosons, or even distinguishable)
with Coulomb interaction.

\begin{remark}\label{rmk:diagonal-set}
	Note that when $2s\ge d$, any wave function in the quadratic form domain 
	of the operator on the left hand side of \eqref{eq:LT_frac} must vanish smoothly 
	on the diagonal set
	$$
		\bDelta := \{(x_i)_{i=1}^N \in (\R^d)^N : x_i=x_j~\text{for some}~i\ne j\}.
	$$
	When $d=s=1$, it is well known \cite{Girardeau-60} that any symmetric wave function vanishing 
	on the diagonal set is equal to an anti-symmetric wave function up to multiplication by an 
	appropriate sign function, and hence \eqref{eq:LT_frac} boils down to a consequence of 
	\eqref{eq:LT_frac_fermions} in this particular case. In higher dimension, this correspondence between bosonic and 
	fermionic wave functions is not available and it is interesting to ask if a Lieb-Thirring 
	inequality of the form \eqref{eq:LT_frac_fermions} holds true for all wave functions vanishing on the diagonal 
	set (without the anti-symmetry assumption). We refer to Section~\ref{sec:proof_constants} 
	for a detailed discussion.
\end{remark}

\begin{remark}
	We have for simplicity fixed the interaction strength in \eqref{eq:LT_frac} to unity. 
	One may consider adding a coupling parameter $\lambda>0$ 
	to the interaction term
	and study the inequality
	\begin{equation} \label{eq:LT_frac_coupling}
			\left\langle \Psi, \left( \sum_{i=1}^N (-\Delta_{i})^s 
			+ \sum_{1\le i<j \le N} \frac{\lambda}{|x_i-x_j|^{2s}} \right) \Psi \right\rangle 
			\ge C(\lambda) \int_{\R^d} \rho_{\Psi}(x)^{1+2s/d} \,dx
	\end{equation}
	for all $N\ge 2$ and all normalized wave functions $\Psi \in H^s(\R^{dN})$,
	with a constant $C(\lambda)$ independent of $N$ and $\Psi$. 
	It is clear that $C(\lambda)>0$ for all $\lambda,s>0$ and $d\ge 1$. 
	However, since the parameter $\lambda$ cannot be removed by scaling, 
	it is interesting to ask for the behavior of the optimal constant of \eqref{eq:LT_frac_coupling}
	in the limits $\lambda \to 0$ and $\lambda \to \infty$. 
	This issue will be thoroughly discussed in Section~\ref{sec:proof_constants}.
\end{remark}

\begin{remark}\label{rem:magnetic}
	When $0<s\leq1$ we can also replace the one-body kinetic operator $(-\Delta)^{s}$ by
	$|i\nabla + A(x)|^{2s}$ with $A \in L_{\loc}^{2}(\R^d;\R^d)$ being a magnetic vector potential.
	By virtue of the
	diamagnetic inequality (see e.g. \cite[Eq.~(2.3)]{FraLieSei-07})
	\begin{align} \label{eq:diam_ineq}
		\langle u, |i\nabla+A|^{2s} u \rangle \geq \langle |u|, (-\Delta)^{s} |u| \rangle 
	\end{align}
	the inequalities \eqref{eq:LT_frac}-\eqref{eq:LT_frac_fermions}-\eqref{eq:LT_frac_coupling}
	hold with the same constants (independent of $A$).
\end{remark}

When $s \notin \N$, the Lieb-Thirring inequality \eqref{eq:LT_frac} 
cannot be obtained from a straightforward modification of the proof 
of \eqref{eq:LT-boson} in \cite{LunPorSol-14}. 
The non-local property of $(-\Delta)^s$ complicates the local uncertainty 
principle and a fractional interpolation inequality on cubes is required. 
We will follow the strategy in \cite{LunPorSol-14}, but several technical 
adjustments are presented. 
The details are provided in Section~\ref{sec:LT}.  
We believe that our 
presentation here provides a unified framework for proving Lieb-Thirring
inequalities by means of local formulations of the uncertainty and exclusion
principles, and
can be used to simplify many parts of the previous works 
\cite{LunSol-13,LunSol-14,FraSei-12,LunPorSol-14}.
For comparison, we also make a note about fermions and weaker exclusion 
principles in Section \ref{sec:fermion}.

\subsection{Hardy-Lieb-Thirring inequality} 
Recall that for every $0<s<d/2$
we have the Hardy inequality\footnote{The case $s \ge d/2$ requires
	additional boundary conditions at $x=0$ and will not be treated here.
	See \cite{Yafaev-99}, and \cite{EkhEnb-10} for corresponding 
	fermionic Lieb-Thirring inequalities.}
\cite{Herbst-77}
$$
	(-\Delta)^s - \mathcal{C}_{d,s} |x|^{-2s} \ge 0
	\quad \text{on~}L^2(\R^d),
$$
where the sharp constant is
$$
	\mathcal{C}_{d,s} := 2^{2s}\left( \frac{\Gamma((d+2s)/4)}{\Gamma((d-2s)/4)} \right)^2.
$$
We will prove the following improvement of Theorem \ref{thm:LT_frac} 
when $0<s<d/2$.

\begin{theorem}[Hardy-Lieb-Thirring inequality] \label{thm:HLT_frac} 
	For all $d\ge 1$ and $0<s<d/2$, there exists a constant $C>0$ depending only on 
	$d$ and $s$ such that for every ($L^2$-normalized) function 
	$\Psi\in H^s(\R^{dN})$ 
	and for all $N\in \N$, we have
	\begin{multline}\label{eq:HLT_frac}
		\left\langle \Psi, \left(  \sum_{i=1}^N \left( (-\Delta_{i})^s - \frac{\mathcal{C}_{d,s}}{|x_i|^{2s}} \right) 
		+ \sum_{1\le i<j\le N} \frac{1}{|x_i-x_j|^{2s}}\right) \Psi \right\rangle\\
		\ge C \int_{\R^d} \rho_\Psi(x)^{1+2s/d}\,dx.
	\end{multline}
\end{theorem}

For $s=1/2$ and $d=3$, the operator in \eqref{eq:HLT_frac} can be interpreted 
as the Hamiltonian of a system of $N$ equally charged relativistic particles
(bosons, fermions or distinguishable)
moving around a static `nucleus' of opposite charge located at $x=0$, where
all particles interact via Coulomb forces.

Our result \eqref{eq:HLT_frac} can be considered as the interacting bosonic 
analogue to the following Hardy-Lieb-Thirring inequality for fermions:
\begin{align} \label{eq:HLT-fermions}
	\left\langle \Psi,  \sum_{i=1}^N \left( (-\Delta_i)^s - \frac{\mathcal{C}_{d,s}}{|x_i|^{2s}} \right) \Psi \right\rangle 
	\ge C \int_{\R^d} \rho_\Psi(x)^{1+2s/d}dx,
\end{align}
which holds for every wave function $\Psi$ satisfying the anti-symmetry 
\eqref{eq:anti-symmetry}. 
The bound \eqref{eq:HLT-fermions} was proved for $s=1$ by  Ekholm and Frank 
\cite{EkhFra-06}, for $0<s\le 1$ by Frank, Lieb and Seiringer \cite{FraLieSei-07}, 
and for $0<s<d/2$ by Frank \cite{Frank-09}. 
In fact, \eqref{eq:HLT-fermions} is dually equivalent to a lower bound 
on the sum of negative eigenvalues of the one-body operator
$(-\Delta)^s - \mathcal{C}_{d,s}|x|^{-2s}+V(x)$ and 
such a bound was
proved in \cite{EkhFra-06,FraLieSei-07,Frank-09}. 
Unfortunately this duality argument 
(which has been the traditional route to proving Lieb-Thirring inequalities)
does not apply in our interacting bosonic case.

\begin{remark}
	The motivation for \eqref{eq:HLT-fermions} was critical stability 
	of relativistic matter in the presence of magnetic fields.
	In both \eqref{eq:HLT_frac} and \eqref{eq:HLT-fermions} we can, for $0 < s \leq 1$, 
	replace $(-\Delta)^s$ with a magnetic operator 
	$|i\nabla + A(x)|^{2s}$; cf. Remark~\ref{rem:magnetic}.
\end{remark}

The proof of \eqref{eq:HLT-fermions} in \cite{Frank-09} is based on the 
following powerful improvement of Hardy's inequality: 
For every $d\ge 1$ and $0<t<s<d/2$, there exists a constant $C>0$ depending 
only on $d,s,t$ such that 
\begin{align} \label{eq:Hardy-improved}
 (-\Delta)^s - \frac{\mathcal{C}_{d,s}}{|x|^{2s}} \ge C \ell^{s-t}(-\Delta)^t -\ell^{s} \quad \text{on}~ L^2(\R^d),\quad \forall \ell>0.
\end{align}
Note that
by taking the expectation against a function $u$ and optimizing over $\ell>0$, 
we can see that \eqref{eq:Hardy-improved} is equivalent to the interpolation 
inequality 
\begin{align} \label{eq:Hardy-improved-interpolation}
\left \langle u, \left( (-\Delta)^s - \frac{\mathcal{C}_{d,s}}{|x|^{2s}} \right) u \right \rangle^{t/s}\left( \int_{\R^d} |u|^2\right)^{1-t/s} \ge C \langle u, (-\Delta)^t u\rangle.
\end{align}
By Sobolev's embedding (see, e.g., 
\cite{LieLos-01,CotTav-04} for the sharp constant)
\begin{align} \label{eq:Sobolev-embedding}
	\langle u, (-\Delta)^t u\rangle \ge C \|u\|_{L^q(\R^d)}^2, 
	\quad q=\frac{2d}{d-2t},
	\quad 0<t<d/2,
\end{align}
the bound \eqref{eq:Hardy-improved-interpolation} implies the 
Gagliardo-Nirenberg type inequality
\begin{align} \label{eq:Hardy-improved-GN}
\left \langle u, \left( (-\Delta)^s - \frac{\mathcal{C}_{d,s}}{|x|^{2s}} \right) u \right \rangle^{t/s}\left( \int_{\R^d} |u|^2\right)^{1-t/s} \ge C \|u\|_{L^q(\R^d)}^2, \quad q=\frac{2d}{d-2t}.
\end{align}
The bound \eqref{eq:Hardy-improved} was first proved for $s=1/2$, $d=3$ 
by Solovej, S\o rensen and Spitzer  \cite[Lemma 11]{SolSorSpi-10} and was 
generalized to the full case $0<s<d/2$ by Frank \cite[Theorem 1.2]{Frank-09}.

In fact, \eqref{eq:Hardy-improved} is also a key ingredient of our 
proof of \eqref{eq:HLT_frac}.
The overall strategy is similar to the proof 
of the fractional Lieb-Thirring inequality \eqref{eq:LT_frac}. 
However, since the system is not translation invariant anymore, 
the local uncertainty becomes much more involved. 
We need to introduce a partition of unity and 
use \eqref{eq:Hardy-improved-interpolation} and \eqref{eq:Hardy-improved-GN} to
control the localization error caused by the non-local operator 
$(-\Delta)^s$. The details will be provided in Section \ref{sec:HLT}.

\subsection{Interpolation inequalities} \label{ssec:interpolation} 
Let us concentrate again on the case $0<s<d/2$. 
By applying the Lieb-Thirring inequality in Theorem~\ref{thm:LT_frac} to the 
product wave function $\Psi = u^{\otimes N}$ with $\|u\|_{L^2(\R^d)}=1$, we obtain  
\begin{align}\label{eq:LT-inter-0}
	N \langle u, (-\Delta)^s u \rangle + \frac{N(N-1)}{2} 
	&\iint_{\R^d \times \R^d} \frac{|u(x)|^2|u(y)|^2}{|x-y|^{2s}}\,dxdy \nn\\
	&\ge C N^{1+2s/d} \int_{\R^d} |u(x)|^{2(1+2s/d)}\,dx.
\end{align}
Since the inequality holds for all $N\in \N$, it then follows that 
\begin{align}\label{eq:LT-inter-2}
	\mu \langle u, (-\Delta)^s u \rangle 
	+ \frac{\mu^2}{2} \iint_{\R^d\times \R^d} \frac{|u(x)|^2|u(y)|^2}{|x-y|^{2s}}\,dxdy \nn\\
	\ge C \mu^{1+2s/d} \int_{\R^d} |u(x)|^{2(1+2s/d)} dx
\end{align}
for all $\mu \geq 1$ (possibly with a smaller constant). 
On the other hand, by using Sobolev's embedding
\eqref{eq:Sobolev-embedding} 
and H\"older's interpolation inequality for $L^p$-spaces, we get
\begin{align} \label{eq:GN-s}
	\langle u, (-\Delta)^s u \rangle \ge C \|u\|^2_{L^{2d/(d-2s)}} \ge C\frac{\int_{\R^d} |u|^{2(1+2s/d)}}{(\int_{\R^d} |u|^2)^{2s/d}} = C\int_{\R^d} |u|^{2(1+2s/d)}
\end{align}
which implies \eqref{eq:LT-inter-2} when $0<\mu<1$. 
Thus \eqref{eq:LT-inter-2} holds for all $\mu>0$, and 
optimizing over $\mu$ gives the 
interpolation inequality
\begin{multline} \label{eq:LT-inter-1}
	\langle u, (-\Delta)^s u \rangle^{1-2s/d} 
	\Big( \iint_{\R^d \times \R^d} \frac{|u(x)|^2|u(y)|^2}{|x-y|^{2s}}\,dxdy \Big)^{2s/d} \\
	\ge  C \int_{\R^d} |u(x)|^{2(1+2s/d)}\,dx
\end{multline} 
for $u\in H^s(\R^d)$, $\|u\|_{L^2}=1$. 
Note that in \eqref{eq:LT-inter-1} the normalization 
$\|u\|_{L^2}=1$ can be dropped by scaling.

The interpolation inequality \eqref{eq:LT-inter-1} was first proved for 
the case $s=1/2,d=3$ by Bellazzini, Ozawa and Visciglia \cite{BelOzaVis-11}, 
and was then generalized to the general case $0<s<d/2$ by 
Bellazzini, Frank and Visciglia \cite{BelFraVis-14}. 
The proofs in \cite{BelFraVis-14,BelOzaVis-11} use fractional calculus 
on the whole space and are very different from our approach using the 
Lieb-Thirring inequality. 

\begin{remark} 
	The inequality \eqref{eq:LT-inter-1} is an end-point case of a series of 
	interpolation inequalities in \cite{BelFraVis-14}. 
	The existence of optimizer in this case is open. 
	If a minimizer exists, by formally analyzing the Euler-Lagrange equation 
	we expect that it belongs to $L^{2+\eps}(\R^d)$ for any $\eps>0$ small, 
	but not $L^2(\R^d)$. Thus \eqref{eq:LT-inter-1} can be interpreted as an 
	energy bound for systems of infinitely many particles.
\end{remark}

\begin{remark}
	Note that, when $s\ge d/2$, one has 
	$$
		\iint_{\R^d \times \R^d} \frac{|u(x)|^2|u(y)|^2}{|x-y|^{2s}}\,dxdy = + \infty 
	$$
	for all $u\ne 0$ since $|x|^{-2s}$ is not locally integrable. 
	Therefore, the interpolation inequality \eqref{eq:LT-inter-1} is trivial in this case. 
	However, the Lieb-Thirring inequality \eqref{eq:LT_frac} is non-trivial 
	for all $s> 0$ because the wave 
	function $\Psi$ may vanish on the diagonal set (see Remark \ref{rmk:diagonal-set}).
\end{remark}

In principle, the implication of a one-body inequality from a many-body 
inequality is not surprising. However, in the following result we show that 
the reverse implication also holds true under certain conditions.   

\begin{theorem}\label{thm:LT_equiv}
	For $0<s<d/2$ and $s \leq 1$, the Lieb-Thirring inequality 
	\eqref{eq:LT_frac} is \emph{equivalent} to 
	the one-body interpolation inequality \eqref{eq:LT-inter-1}.
\end{theorem}

As we explained above, the implication of \eqref{eq:LT-inter-1}
from \eqref{eq:LT_frac} works for all $0<s<d/2$.
The implication of \eqref{eq:LT_frac} from \eqref{eq:LT-inter-1} 
is more subtle and we obtain it from
fractional versions of the Hoffmann-Ostenhof inequality \cite{Hof-77}, 
which requires $0<s\leq1$, and a generalized version of the 
Lieb-Oxford inequality \cite{Lieb-79,LieOxf-80} for homogeneous potentials.
We will provide these details in Section~\ref{sec:LT-interpolation}.

\begin{remark}
	Unfortunately, we can not offer an exact relation between the optimal constants in 
	\eqref{eq:LT_frac} and \eqref{eq:LT-inter-1}. 
	On the other hand, from \eqref{eq:LT-inter-0} it is obvious that the 
	optimal constant in \eqref{eq:LT_frac} is not bigger than 
	the optimal constant $C_1$ in the inequality
	\begin{align*} 
		\langle f, (-\Delta)^s f \rangle + \frac{1}{2} 
		\iint_{\R^d \times \R^d} \frac{|f(x)|^2|f(y)|^2}{|x-y|^{2s}}\,dxdy \ge C_1 \int_{\R^d} |f(x)|^{2(1+2s/d)}\,dx.
	\end{align*}
	for all $f\in H^s(\R^d)$ (not necessarily normalized), which is related to the optimal 
	constant $C$ in \eqref{eq:LT-inter-1} by the exact formula 
	$$
		\frac{C_1}{C}=\inf_{t>0} \left( 1+ \frac{t}{2} \right)t^{-2s/d} 
		= \left( 1- \frac{2s}{d} \right)^{-1+2s/d} \left( \frac{d}{4s} \right)^{2s/d}.
	$$
\end{remark}

By the same proof as that
of Theorem~\ref{thm:LT_equiv}, we also obtain the 
following equivalence for the Hardy-Lieb-Thirring inequality \eqref{eq:HLT_frac}.

\begin{theorem}\label{thm:HLT_equiv}
	For $0<s<d/2$ and $s \leq 1$, the Hardy-Lieb-Thirring inequality 
	\eqref{eq:HLT_frac} is equivalent to the one-body interpolation inequality
	\begin{align} \label{eq:HLT-inter-1}
		\left\langle u, \Big((-\Delta)^s - \mathcal{C}_{d,s} |x|^{-2s} \Big)  u \right\rangle^{1-2s/d} 
		& \left( \iint_{\R^d \times \R^d} \frac{|u(x)|^2|u(y)|^2}{|x-y|^{2s}}\,dxdy \right)^{2s/d} \nn \\
		&\ge C \int_{\R^d} |u(x)|^{2(1+2s/d)}\,dx.
	\end{align}
\end{theorem}

The interpolation inequality \eqref{eq:HLT-inter-1} seems to be new. 
Note that the implication of \eqref{eq:HLT-inter-1} from \eqref{eq:HLT_frac} 
holds for all $0<s<d/2$ (by exactly the same argument as above), 
and hence \eqref{eq:HLT-inter-1} is also valid in this maximal 
range. There might be some way to prove \eqref{eq:HLT-inter-1} directly 
(as in the proof of \eqref{eq:LT-inter-1} in \cite{BelOzaVis-11,BelFraVis-14}), 
but we have not found such a proof yet. 

Finally, we mention that our approach in this paper can be used to prove many 
other interpolation inequalities which do not really come from many-body 
quantum theory. For example, we have 

\begin{theorem}[Isoperimetric inequality with non-local term] \label{thm:iso}
	For any $d \geq 2$ and $1/2 \le s < d/2$
	there exists a constant $C>0$ depending only on $d$ and $s$, 
	such that for all functions $u \in W^{1,2s}(\R^d)$ we have
	\begin{multline}\label{eq:LT-inter-3}
		\left( \int_{\R^d} |\nabla u|^{2s} dx \right)^{1-2s/d}
		\left(\iint_{\R^d \times \R^d} \frac{|u(x)|^{2s} |u(y)|^{2s}}{|x-y|^{2s}} dxdy\right)^{2s/d} \\
		\ge C \int_{\R^d} |u|^{2s(1+2s/d)}\,dx.
	\end{multline} 
\end{theorem}

This inequality seems to be new and it could be useful in the 
context of isoperimetric inequalities with competing non-local term; 
see \cite[Lemma~7.1]{KnuMur-14}, \cite[Lemma~5.2]{KnuMur-13} 
and \cite[Lemma~B.1]{Muratov-10} for relevant results. 
The proof of Theorem \ref{thm:iso} will be given in Section~\ref{sec:LT-interpolation}.

\section{Fractional Lieb-Thirring inequality}\label{sec:LT}
In this section we prove the fractional Lieb-Thirring inequality \eqref{eq:LT_frac}.
We shall follow the overall strategy in \cite{LunPorSol-14}, 
where we localize the interaction and kinetic energies into disjoint cubes, 
but we also introduce several new tools.

\subsection{Local exclusion}
The following result is a simplified version of the local exclusion principle in 
\cite[Theorem 2 and Section~4.2]{LunPorSol-14}.

\begin{lemma}[Local exclusion]\label{lem:local-ex} 
	For all $d\geq 1$, $s > 0$, for every normalized function $\Psi\in L^2(\R^{dN})$ and 
	for an arbitrary collection of disjoint cubes $Q$'s in $\R^d$, one has
	\begin{multline}\label{eq:local-exclusion}
		\left\langle \Psi , \sum_{1\le i<j \le N} \frac{1}{|x_i-x_j|^{2s}} \Psi \right \rangle
		\ge \sum_Q \frac{1}{2d^s|Q|^{2s/d}} \left[ \Big( \int_Q \rho_\Psi \Big)^2 - \int_Q \rho_\Psi \right]_+.
	\end{multline}
\end{lemma}
\begin{proof} 
	The following argument goes back to Lieb's work on the indirect energy \cite{Lieb-79}.
	Since the interactions between different cubes are positive and $|x-y|\le \sqrt{d}|Q|^{1/d}$ for all $x,y\in Q$, 
	we have
	\begin{align*}  
		\sum_{1\le i<j \le N} \frac{1}{|x_i-x_j|^{2s}} 
		&\ge \sum_Q \sum_{1\le i<j \le N} \frac{\1_Q(x_i)\1_Q(x_j)}{|x_i-x_j|^{2s}}\\
		&\ge \sum_Q \frac{1}{d^s|Q|^{2s/d}} \sum_{1\le i<j \le N} \1_Q(x_i)\1_Q(x_j)\\
		&= \sum_Q \frac{1}{2d^s|Q|^{2s/d}} \left[\left(\sum_{i=1}^N \1_Q(x_i) \right)^2 - \sum_{i=1}^N \1_Q(x_i) \right]_+.
	\end{align*}
	Taking the expectation against $\Psi$ and using the Cauchy-Schwarz inequality 
	$$
		\left\langle \Psi, \Big(\sum_{i=1}^N \1_Q(x_i) \Big)^2 \Psi \right \rangle 
		\ge \left\langle \Psi, \sum_{i=1}^N \1_Q(x_i)  \Psi \right \rangle ^2 = \left( \int_Q \rho_\Psi \right)^2,
	$$
	we obtain the desired estimate.
\end{proof}

\subsection{Local uncertainty}
Now we localize the kinetic energy into disjoint cubes $Q$'s.
For every $s>0$ we can write $s=m+\sigma$ with $m\in \{0,1,2,\dots\}$ 
and $0 \le \sigma<1$. 
Then for any one-body function $u \in H^s(\R^d)$ we have
\begin{align*}
	\langle u, (-\Delta)^{s} u \rangle 
	&= \int_{\R^d} |p|^{2s} |\widehat u (p)|^2 dp 
	= \int_{\R^d} |p|^{2\sigma} \Big( \sum_{i=1}^d {p_i}^2\Big)^m |\widehat u (p)|^2 dp \nn\\
	&= \sum_{|\alpha|=m} \frac{m!}{\alpha!} \int_{\R^d} |p|^{2\sigma} \prod_{i=1}^d p_i^{2\alpha_i} |\hat u (p)|^2  dp \nn\\
	&= \sum_{|\alpha|=m} \frac{m!}{\alpha!} \langle D^\alpha u, (-\Delta)^\sigma  D^\alpha u \rangle .
\end{align*}
The last sum is taken over multi-indices $\alpha=(\alpha_1,\dots,\alpha_d)\in \{0,1,2,\dots\}^d$ with
$$
	|\alpha|=\sum_{i=1}^d \alpha_i, \quad \alpha! =\prod_{i=1}^d (\alpha_i !) \quad \text{and}\quad 
	D^\alpha = \prod_{i=1}^d \frac{\partial^{\alpha_i}}{\partial_{r_i}^{\alpha_i}}.
$$
Here we denoted by $p=(p_1,p_2,\dots,p_d)\in \R^d$ and $r=(r_1,\dots,r_d)\in \R^d$, the variables in the Fourier space and the configuration space, respectively.  

If $s=m$, we have
\bq \label{eq:s=m}
	\langle u, (-\Delta)^{s} u \rangle = \sum_{|\alpha|=m} \frac{m!}{\alpha!} \int_{\R^d} |D^\alpha u| 
	\ge \sum_{|\alpha|=m} \frac{m!}{\alpha!} \sum_Q \int_Q |D^\alpha u|
\eq
for disjoint cubes $Q$'s. On the other hand, if $m<s<m+1$, 
then using the quadratic form
representation\footnote{Note that this formula only holds for $0<\sigma<1$.} 
(see, e.g., \cite[Lemma~3.1]{FraLieSei-07})
\begin{align} \label{eq:fractional-form}
	\langle f, (-\Delta)^{\sigma} f\rangle = c_{d,\sigma}\int_{\R^d}\int_{\R^d}\frac{|f(x)-f(y)|^2}{|x-y|^{d+2\sigma}}\,dx dy,
\end{align}
where
$$
c_{d,\sigma} :=\frac{2^{2\sigma-1}}{\pi^{d/2}} \frac{\Gamma((d+2\sigma)/2)}{|\Gamma(-\sigma)|},
$$
we have
\begin{align} \label{eq:s>m}
	\langle u, (-\Delta)^{s} u \rangle 
	&= c_{d,\sigma} \sum_{|\alpha|=m} \frac{m!}{\alpha!} \int_{\R^d \times \R^d} 
	\frac{|D^\alpha u (x) -D^\alpha u (y)|^2}{|x-y|^{d+2\sigma}} dxdy \nn\\
	& \ge c_{d,\sigma} \sum_{|\alpha|=m} \frac{m!}{\alpha!} \sum_Q \int_{Q \times Q} 
	\frac{|D^\alpha u (x) -D^\alpha u (y)|^2}{|x-y|^{d+2\sigma}} dxdy
\end{align}
for disjoint cubes $Q$'s. It is convenient to combine \eqref{eq:s=m} and \eqref{eq:s>m} into a single formula
\begin{align}\label{eq:u-Hs-Ts}
	\langle u, (-\Delta)^{s} u \rangle \ge \sum_Q \norm{u}_{\dot{H}^s(Q)}^2,
\end{align}
where the semi-norm $\norm{u}_{\dot{H}^s(Q)}^2$ of $u\in L^2(Q)$ on a cube $Q$ is defined by
\[
	\norm{u}_{\dot{H}^s(Q)}^2:= \left\{\begin{array}{ll}
	\sum\limits_{|\alpha|=m} \dfrac{m!}{\alpha!} \int_{Q} |D^\alpha u|^2, &\quad \text{if }s=m,\\ 
	c_{d,\sigma}\sum\limits_{|\alpha|=m} \dfrac{m!}{\alpha!} \iint_{Q\times Q} \dfrac{|D^\alpha u (x) 
	-D^\alpha u (y)|^2}{|x-y|^{d+2\sigma}}\,dxdy, &\quad \text{if } 0<\sigma<1.
	\end{array} \right.
\]

The following estimate plays an essential role in our proof.
\begin{lemma}[Local uncertainty]\label{lem:local_uncert}
	For every $d\geq 1$, $s>0$, cube $Q\subset \R^d$ and $u\in L^2(Q)$, one has
	\bq \label{eq:local-uncertainty}
		\norm{u}_{\dot{H}^s(Q)}^2  
		\ge \frac{1}{C}\,\frac{\int_Q |u|^{2(1+2s/d)}}{\Big(\int_Q |u|^2 \Big)^{2s/d}} - \frac{C}{|Q|^{2s/d}} \int_Q |u|^2
	\eq
	for a constant $C>0$ independent of $Q$ and $u$.
\end{lemma}
Before proving Lemma~\ref{lem:local_uncert}, let us clarify a technical point 
concerning the Sobolev space $H^s(Q)=W^{s,2}(Q)$, whose intrinsic norm can be defined by 
(see e.g. \cite[Section 7.36 and Theorem 7.48]{Adams-75})
$$
	\|u\|_{H^s(Q)}^2:= \norm{u}_{\dot{H}^s(Q)}^2 + \sum_{|\alpha|\le m} \int_Q |D^\alpha u|^2.
$$
Here recall that $s=m+\sigma$ with $m\in \{0,1,2,\dots\}$ and $0\le \sigma<1$. 
By Poincar\'e's inequality for $\dot H^\sigma(Q)$ 
(see, e.g., \cite[Lemma 2.2]{HurVah-13}) and the elementary inequality 
$|a-b|^2 \ge \frac{1}{2} |a|^2 - |b|^2 $ for $a,b\in \C$, we have
\begin{equation*} 
	C \norm{u}_{\dot{H}^s(Q)}^2 
	\ge \sum_{|\alpha|=m} \Big\| D^\alpha u 
		- \frac{1}{|Q|} \int_Q D^\alpha u \Big\|_{L^2(Q)}^2 
	\ge \frac{1}{2}\| D^\alpha u \|_{L^2(Q)}^2 - \frac{\left|\int_Q D^\alpha u\right|^2}{|Q|}.
\end{equation*}  
From the latter estimate and Sobolev's embedding, it is straightforward 
to obtain the following equivalence of norms 
\begin{align}\label{eq:equiv-norm-Sobolev}
		\|u\|_{H^s(Q)}^2 \geq \norm{u}_{\dot{H}^s(Q)}^2+ \int_Q |u|^2 \ge C_Q \|u\|_{H^s(Q)}^2,
\end{align}
for a constant $C_Q>0$ depending only on the the cube $Q$. Now we provide

\begin{proof}[Proof of Lemma~\ref{lem:local_uncert}]
	By translating and dilating, that is, replacing $u(x)$ by $u(\lambda(x-x_0))$ for $\lambda>0$ 
	and $x_0\in \R^d$, it suffices to consider the unit cube $Q=[0,1]^d$. 
	Then, thanks to \eqref{eq:equiv-norm-Sobolev}, it remains to prove the 
	fractional Gagliardo-Niren\-berg inequality
	\bq\label{eq:local-uncertainty-Hs}
		\|u\|_{H^s(Q)}^{\theta} \|u\|_{L^2(Q)}^{1-\theta} \ge C \|u\|_{L^{q}(Q)},\quad q=2+ \frac{4s}{d},\quad  \theta=\frac{d}{d+2s}
	\eq 
	for a constant $C>0$ independent of $u$. Since the (unit) cube $Q$ is regular, 
	we may apply the extension theorem to $H^s(Q)$ 
	(see \cite[Theorem~7.41]{Adams-75} or \cite[Theorem~4.2.3]{Triebel-78}) 
	and obtain for any function $u\in H^s(Q)$ a function 
	$U\in H^s(\R^d)$ satisfying
	$$
		U|_{Q}= u, \quad  \|U\|^2_{L^2(\R^d)} \le C \|u\|_{L^2(Q)}, \quad \|U\|_{H^s(\R^d)} \le C \|u\|_{H^s(Q)},
	$$
	where $C>0$ depends only on $d$ and $s$. We will show that 
	\bq \label{eq:fractional-GN-Rd}
		\|U\|_{\dot H^s(\R^d)}^{\theta} \|U\|_{L^2(\R^d)}^{1-\theta} 
		\ge C \|U\|_{L^{q}(\R^d)},\quad q=2+\frac{4s}{d},\quad\theta=\frac{d}{d+2s},
	\eq
	and \eqref{eq:local-uncertainty-Hs} follows immediately. 
	By Sobolev's embedding \eqref{eq:Sobolev-embedding} 
	$$ \|U\|_{\dot H^{\theta s}(\R^d)}  \ge C\|U\|_{L^{q}(\R^d)}, \quad q=2+\frac{4s}{d}=\frac{2d}{d-2\theta s},$$
	the estimate \eqref{eq:fractional-GN-Rd} follows from the following 
	interpolation inequality
	\bq \label{eq:fractional-GN-Rd-1}
		\|U\|_{\dot H^s(\R^d)}^{\theta} \|U\|_{L^2(\R^d)}^{1-\theta} 
		\ge \|U\|_{\dot H^{\theta s}(\R^d)},\quad \forall \theta \in (0,1),
	\eq
	which is in turn a simple consequence of H\"older's inequality
	$$
		\left( \int_{\R^d} p^{2s} |\widehat{U}(p)|^2 dp\right)^{\theta} \left( \int_{\R^d} |\widehat{U}(p)|^2 dp \right)^{1-\theta} 
		\ge \int_{\R^d} p^{2\theta s} |\widehat{U}(p)|^2 dp .
	$$
\end{proof}

\begin{remark}
	Note that to the semi-norm $\norm{\,\cdot\,}_{\dot{H}^s(\Omega)}$
	there is a naturally associated operator,
	which for $s=1$ coincides with $-\Delta_\Omega^{\cN}$, 
	the Neumann Laplacian on $\Omega \subseteq \R^d$.
	It is a relevant question whether for $0<s<1$ 
	and bounded domains $\Omega$
	this operator coincides with 
	$(-\Delta_\Omega^{\cN})^s$ (defined using the spectral theorem),
	something that was 
	shown in \cite{FraGei-14} to be false
	in the case of the Dirichlet Laplacian $-\Delta_\Omega^{\cD}$
	(see also \cite{MusNaz-14,SerVal-14,Grubb-14} for related results).
	In any case, the analogue of \eqref{eq:local-uncertainty} 
	for $(-\Delta_\Omega^{\cN/\cD})^s$
	can be proved using the method in \cite{Rum-11}.
\end{remark}

We will need the following many-body version of 
Lemma~\ref{lem:local_uncert}. 
\begin{lemma}[Many-body version of local uncertainty]\label{lem:many-body-local-uncertainty}
	For any $L^2$-normalized function 
	$\Psi\in H^s(\R^{dN})$ and for an arbitrary 
	collection of disjoint cubes $Q$'s,
	the kinetic energy satisfies the estimate
	\begin{align}\label{eq:many-body-local-uncertainty}
		\left\langle \Psi, \sum_{i=1}^N (-\Delta_{i})^s \Psi \right\rangle 
		\ge \sum_Q \left[\frac{1}{C}\frac{\int_Q \rho_\Psi^{1+2s/d}}{\Big(\int_Q \rho_\Psi \Big)^{2s/d}} 
		- \frac{C}{|Q|^{2s/d}} \int_Q \rho_\Psi \right],
	\end{align}
	where $C$ is the same constant as in Lemma~\ref{lem:local_uncert}.
\end{lemma}

\begin{proof} 
	Let $\gamma_\Psi^{(1)}$ be the one-body density matrix of $\Psi$ 
	(see \cite[Section 3.1.5]{LieSei-10}), which is a non-negative 
	trace class operator on $L^2(\R^d)$ with kernel 
	\begin{align} \label{eq:1pdm}
		\gamma_\Psi^{(1)}(x,y) := \sum_{j=1}^N \int_{\R^{d(N-1)}} &\Psi(x_1,\dots,x_{j-1},x,x_{j+1},\dots,x_N) \times \nn\\
		&\times \overline{\Psi(x_1,\dots,x_{j-1},y,x_{j+1},\dots,x_N)}\, \prod\limits_{i\ne j} dx_i.
	\end{align}
	Since $\gamma_\Psi^{(1)}$ is trace class, we can write  
	$$
		\gamma_\Psi^{(1)}(x,y)= \sum_{n\ge 1} u_n(x) \overline{u_n(y)},
	$$
	where $u_n \in L^2(\R^d)$ are not necessarily normalized. 
	Then $\rho_\Psi=\sum_{n\ge 1}|u_n|^2$ and 
	\begin{align} \label{eq:Psi-Hs-Ts}
		\left\langle \Psi, \sum_{i=1}^N (-\Delta_{i})^s   \Psi \right\rangle &= \Tr\left[(-\Delta)^s \gamma_\Psi^{(1)}\right] \nn\\
		&=\sum_{n\ge 1} \langle u_n, (-\Delta)^s u_n \rangle \ge \sum_{n\ge 1} \sum_Q \norm{u_n}_{\dot{H}^s(Q)}^2,
	\end{align}
	where we have used \eqref{eq:u-Hs-Ts} in the last estimate. 
	On the other hand, from the local uncertainty \eqref{eq:local-uncertainty}
	we have
	\begin{multline*}
		\left(\int_Q |u_n|^2 \right)^{\frac{2s}{d+2s}} \left( \norm{u_n}_{\dot{H}^s(Q)}^2 
		+ \frac{C}{|Q|^{2s/d}} \int_Q |u_n|^2 \right)^{\frac{d}{d+2s}}\\
		\ge C^{-d/(d+2s)}\| |u_n|^2 \|_{L^{1+2s/d}(Q)}
	\end{multline*}
	for all $n\ge 1$. 
	Therefore, by H\"older's inequality (for sums) and the triangle 
	inequality we get
	\begin{align*}
		&\left( \int_Q \rho_\Psi  \right)^{\frac{2s}{d+2s}}\left( \sum_{n\ge 1} \norm{u_n}_{\dot{H}^s(Q)}^2 
		+ \frac{C}{|Q|^{2s/d}} \int_Q \rho_\Psi  \right)^{\frac{d}{d+2s}}\\
		&\quad= \left( \sum_{n\ge 1} \int_Q |u_n|^2  \right)^{\frac{2s}{d+2s}} \left( \sum_{n\ge 1} \left[ \norm{u_n}_{\dot{H}^s(Q)}^2 
		+ \frac{C}{|Q|^{2s/d}} \int_Q |u_n|^2 \right] \right)^{\frac{d}{d+2s}}  \\
	  	&\quad\ge \sum_{n\ge 1} \left( \int_Q |u_n|^2  \right)^{\frac{2s}{d+2s}} \left( \norm{u_n}_{\dot{H}^s(Q)}^2 
		+ \frac{C}{|Q|^{2s/d}} \int_Q |u_n|^2 \right)^{\frac{d}{d+2s}} \\
		&\quad\ge \sum_{n\ge 1} C^{-\frac{d}{d+2s}} \| |u_n|^2 \|_{L^{1+2s/d}(Q)} 
		\ge  C^{-\frac{d}{d+2s}}\big\| \sum_{n\ge 1} |u_n|^2 \big\|_{L^{1+2s/d}(Q)} \\
		&\quad= C^{-\frac{d}{d+2s}} \big\| \rho_\Psi \big\|_{L^{1+2s/d}(Q)},
	\end{align*}
	which is equivalent to  
	$$
		\sum_{n\ge 1} \norm{u_n}_{\dot{H}^s(Q)}^2 
		\ge \frac{1}{C}\frac{\int_Q \rho_\Psi^{1+2s/d}}{\Big(\int_Q \rho_\Psi \Big)^{2s/d}} - \frac{C}{|Q|^{2s/d}} \int_Q \rho_\Psi .
	$$
	The latter estimate and \eqref{eq:Psi-Hs-Ts} imply the desired inequality 
	\eqref{eq:many-body-local-uncertainty}. 
\end{proof}

\begin{remark} 
	By using the interpolation inequality \eqref{eq:GN-s} and the 
	same argument of the proof of Lemma \ref{lem:many-body-local-uncertainty} 
	(in this case one can work on the whole $\R^d$ and no partition of cubes 
	is needed), 
	we obtain the following generalization of \eqref{eq:LT-fermion-uN}: 
	\begin{align}\label{eq:LT-proof-N-small}
		\left\langle \Psi, \sum_{i=1}^N (-\Delta_{i})^s \Psi \right\rangle 
		\ge C N^{-2s/d} \int_{\R^d} \rho_\Psi^{1+2s/d}
	\end{align}
	for all normalized functions $\Psi\in H^s(\R^{dN})$ 
	and for a constant $C>0$ depending only on $d$ and $s$. 
	When $0<s \le 1$, \eqref{eq:LT-proof-N-small} can also be proved using the 
	Hoffmann-Ostenhof inequality in Lemma \ref{lem:HO} and Sobolev's embedding. 
	We will use \eqref{eq:LT-proof-N-small} to obtain the Lieb-Thirring 
	inequality \eqref{eq:LT_frac} when $N$ is small. 
\end{remark}

\subsection{A covering lemma}
To combine the local uncertainty and exclusion principles, 
we need a nice choice of the partition of cubes $Q$'s. 
The following result is inspired by the work of Lundholm and Solovej \cite[Theorem~11]{LunSol-13}. 
In fact, a similar result can be obtained by following their construction. 
However, our construction below is simpler to apply and 
results in improved constants. 
\begin{lemma}[Covering lemma] \label{lem:covering}
	Let $Q_0$ be a cube in $\R^d$ and let $\Lambda>0$. Let $0\le f\in L^1(Q_0)$ satisfy 
	$\int_{Q_0} f \ge \Lambda>0$. Then $Q_0$ can be divided into disjoint 
	sub-cubes $Q$'s such that:
	
	\begin{itemize}
	\item For all $Q$,
	$$
		\int_{Q} f < \Lambda.
	$$
	\item For all $\alpha>0$ and integer $k \ge 2$ 
	\bq \label{eq:bound-f*f-f}
		\sum_{Q} \frac{1}{|Q|^{\alpha}} \left[ \left(\int_{Q} f \right)^2 
		- \frac{\Lambda}{a} \int_{Q} f \right] \ge 0,
	\eq
	where 
	$$
		a:= \frac{k^d}{2}\left(1 + \sqrt{1 + \frac{1-k^{-d}}{k^{d\alpha}-1} } \right).
	$$
	\item If $k=3$, then the center of $Q_0$ coincides with 
	the center of
	exactly one sub-cube $Q$,
	and the distance from every other sub-cube $Q$ to the center of $Q_0$
	is not smaller than $|Q|^{1/d}/2$.
	\end{itemize}
\end{lemma}

Note that the simplest choice is $k=2$ and it is sufficient for the 
proof of the Lieb-Thirring inequality \eqref{eq:LT_frac}. 
However the case $k=3$ will be more useful for the proof of the 
Hardy-Lieb-Thirring inequality \eqref{eq:HLT_frac} in Section~\ref{sec:HLT}.

\begin{proof}
	First, we divide $Q_0$ into $k^d$ disjoint sub-cubes with $1/k$ 
	of the original side length. 
	For every sub-cube, if the integral of $f$ over it is less than $\Lambda$, 
	then we will not divide it further; 
	otherwise we divide this sub-cube into $k^d$ disjoint smaller cubes with 
	$1/k$ of the side length, and then iterate the process. 
	Since $f$ is integrable, the procedure must stop after finitely many steps 
	and we obtain a division of $Q_0$ into finitely 
	many sub-cubes $Q$'s.

	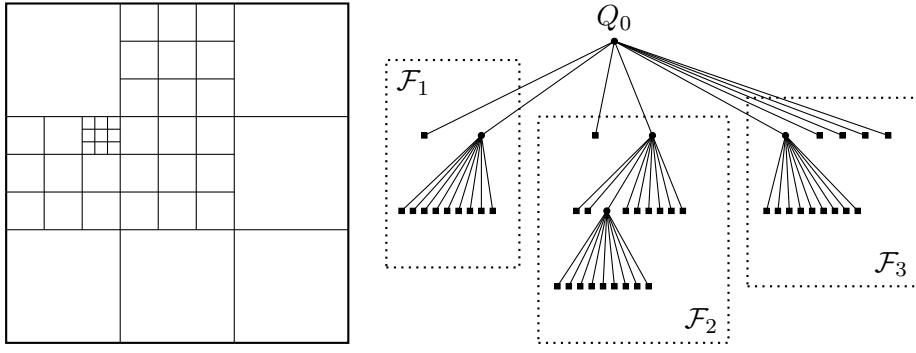
\begin{figure}
		\begin{tikzpicture}
			\draw[thick] (0,0) rectangle (4.5, 4.5);
			
			\draw (1.5,0) -- (1.5,4.5);
			\draw (3,0) -- (3,4.5);
			\draw (0,1.5) -- (4.5,1.5);
			\draw (0,3) -- (4.5,3);
			
			\draw (0,2) -- (3,2);
			\draw (0,2.5) -- (3,2.5);
			\draw (2,1.5) -- (2,4.5);
			\draw (2.5,1.5) -- (2.5,4.5);
			\draw (0.5,1.5) -- (0.5,3);
			\draw (1,1.5) -- (1,3);
			\draw (1.5,3.5) -- (3,3.5);
			\draw (1.5,4) -- (3,4);
			
			\draw (1.166,2.5) -- (1.166,3);
			\draw (1.333,2.5) -- (1.333,3);
			\draw (1,2.666) -- (1.5,2.666);
			\draw (1,2.833) -- (1.5,2.833);
			
			\draw [fill] (8,4) circle [radius = 0.04];
			\node [above] at (8,4) {$Q_0$};
			\draw (8,4) -- (5.5,2.75);
			\draw (8,4) -- (6.25,2.75);
			\draw (8,4) -- (7.75,2.75);
			\draw (8,4) -- (8.5,2.75);
			
			\draw [dotted, thick] (5,1) rectangle (6.75,3.75);
			\node [below right] at (5,3.75) {$\cF_1$};
			\filldraw ([xshift=-1pt,yshift=-1pt]5.5,2.75) rectangle ++(2pt,2pt);
			\draw [fill] (6.25,2.75) circle [radius = 0.04];
			\foreach \x in {0,...,8}
				\filldraw ([xshift=-1pt,yshift=-1pt]5.2+0.15*\x,1.75) rectangle ++(2pt,2pt);
			\foreach \x in {0,...,8}
				\draw (6.25,2.75) -- (5.2+0.15*\x,1.75);
			
			\draw [dotted, thick] (7,0) rectangle (9.5,3);
			\node [above left] at (9.5,0) {$\cF_2$};
			\filldraw ([xshift=-1pt,yshift=-1pt]7.75,2.75) rectangle ++(2pt,2pt);
			\draw [fill] (8.5,2.75) circle [radius = 0.04];
			
			\filldraw ([xshift=-1pt,yshift=-1pt]7.5,1.75) rectangle ++(2pt,2pt);
			\filldraw ([xshift=-1pt,yshift=-1pt]7.65,1.75) rectangle ++(2pt,2pt);
			\draw [fill] (7.9,1.75) circle [radius = 0.04];
			\filldraw ([xshift=-1pt,yshift=-1pt]8.15,1.75) rectangle ++(2pt,2pt);
			\filldraw ([xshift=-1pt,yshift=-1pt]8.3,1.75) rectangle ++(2pt,2pt);
			\filldraw ([xshift=-1pt,yshift=-1pt]8.45,1.75)  rectangle ++(2pt,2pt);
			\filldraw ([xshift=-1pt,yshift=-1pt]8.6,1.75) rectangle ++(2pt,2pt);
			\filldraw ([xshift=-1pt,yshift=-1pt]8.75,1.75) rectangle ++(2pt,2pt);
			\filldraw ([xshift=-1pt,yshift=-1pt]8.9,1.75) rectangle ++(2pt,2pt);
			
			\draw (8.5,2.75) -- (7.5,1.75);
			\draw (8.5,2.75) -- (7.65,1.75);
			\draw (8.5,2.75) -- (7.9,1.75);
			\draw (8.5,2.75) -- (8.15,1.75);
			\draw (8.5,2.75) -- (8.3,1.75);
			\draw (8.5,2.75) -- (8.45,1.75);
			\draw (8.5,2.75) -- (8.6,1.75);
			\draw (8.5,2.75) -- (8.75,1.75);
			\draw (8.5,2.75) -- (8.9,1.75);
			
			\foreach \x in {0,...,8}
				\filldraw ([xshift=-1pt,yshift=-1pt]7.25+0.15*\x,0.75) rectangle ++(2pt,2pt);
			\foreach \x in {0,...,8}
				\draw (7.9,1.75) -- (7.25+0.15*\x,0.75);
			
			\draw [dotted, thick] (9.75,0.75) rectangle (12,3.25);
			\node [above left] at (12,0.75) {$\cF_3$};
			\draw [fill] (10.25,2.75) circle [radius = 0.04];
			\draw (8,4) -- (10.25,2.75);
			\foreach \x in {1,...,4}
				\filldraw ([xshift=-1pt,yshift=-1pt]10.4+0.3*\x,2.75) rectangle ++(2pt,2pt);
			\foreach \x in {1,...,4}
				\draw (8,4) -- (10.4+0.3*\x,2.75);
			\foreach \x in {0,...,8}
				\filldraw ([xshift=-1pt,yshift=-1pt]10+0.15*\x,1.75) rectangle ++(2pt,2pt);
			\foreach \x in {0,...,8}
				\draw (10.25,2.75) -- (10+0.15*\x,1.75);			
		\end{tikzpicture}
		\caption{Example of a division of $Q_0$ (in $d=2$) with $k=3$.}
		\label{fig:splitting}
	\end{figure}

	It is obvious that for every sub-cube $Q$ one has  $\int_{Q} f < \Lambda$ 
	and $|Q|=k^{-\ell(Q)d}|Q_0|$ for some level $\ell(Q)\in \{0,1,2,\dots\}$. 
	By viewing the sub-cubes as the leaves of a full $k^d$-ary tree corresponding to the 
	above division (cf. Figure~\ref{fig:splitting}),
	we can distribute all sub-cubes into disjoint 
	groups $\{\cF_i\}$ such that in each group $\cF_i$:
	\begin{itemize}
		\item There are exactly $k^d$ smallest sub-cubes within $\cF_i$. 
	
		\item The integral of $f$ over the union of these $k^d$ smallest sub-cubes is greater than $\Lambda$.
	
		\item There are at most $(k^d-1)$ sub-cubes of every other volume.
	\end{itemize}
	Now we consider each group $\cF_i$. 
	Let $m_i=\inf_{Q\in \cF_i}|Q|$ denote the minimal volume occuring 
	in the group. 
	By the Cauchy-Schwarz inequality we have
	\begin{align} \label{eq:Fi-smallest}
		&\sum_{Q\in \cF_i, |Q|=m_i} \frac{1}{|Q|^\alpha} \left[ \left( \int_Q f \right) ^2 
		- \frac{\Lambda}{a} \int_Q f \right]\nn\\
		&\quad \ge \frac{1}{m_i^\alpha} \left[ \frac{1}{k^d} \left( \sum_{Q\in \cF_i, |Q|=m_i} \int_Q f \right) ^2 
		- \frac{\Lambda}{a} \sum_{Q\in \cF_i, |Q|=m_i} \int_Q f \right] \nn\\
		&\quad \ge \frac{1}{m_i^\alpha} \left( \frac{\Lambda^2}{k^d} - \frac{\Lambda^2}{a}\right).
	\end{align}
	Here in the last inequality we have used the lower bound
	$$
		\sum_{Q\in \cF_i, |Q|=m_i} \int_Q f \ge \Lambda > \frac{k^d \Lambda}{2a}
	$$
	and that the function $t\mapsto t^2/k^d - (\Lambda/a)t$ is increasing 
	when $t\ge k^d \Lambda/(2a)$.
	On the other hand, using the obvious lower bound
	$$
		\Big(\int_{Q} f \Big)^2- \frac{\Lambda}{a} \int_{Q} f \ge - \frac{\Lambda^2}{4a^2},
	$$
	we find that
	\begin{multline}\label{eq:Fi-larger}
		\sum_{Q\in \cF_i, |Q|>m_i} \frac{1}{|Q|^\alpha} \left[ \left( \int_Q f \right) ^2 
		- \frac{\Lambda}{a}  \int_Q f \right] 
		\ge - \frac{\Lambda^2}{4a^2}\sum_{Q\in \cF_i, |Q|>m_i} \frac{1}{|Q|^\alpha}\\
		\ge - \frac{\Lambda^2}{4a^2} \sum_{j \ge 1} \frac{k^d-1}{ (k^{dj}m_i)^\alpha} 
		= - \frac{\Lambda^2}{4a^2} \frac{k^d-1}{(k^{d\alpha}-1)m_i^\alpha}.
	\end{multline}
	Here in the second inequality we have used the fact that in $\cF_i$, 
	each sub-cube has volume $k^{dj} m_i$ for some $j \in \{0,1,2,\dots\}$ 
	and there are at most $(k^d-1)$ sub-cubes 
	of every volume larger than $m_i$. 
	Adding \eqref{eq:Fi-smallest} and \eqref{eq:Fi-larger}, we find that
	\begin{align*} 
		\sum_{Q\in \cF_i} \frac{1}{|Q|^\alpha} \left[ \left( \int_Q f \right) ^2 - \frac{\Lambda}{a}  \int_Q f \right] 
		\ge \frac{\Lambda^2}{m_i^\alpha} \left( \frac{1}{k^d} - \frac{1}{a} - \frac{k^d-1}{4a^2(k^{d\alpha}-1)}\right) =0,
	\end{align*} 
	where the last identity follows from the choice of $a$. Since the latter 
	inequality holds true for every group $\cF_i$, the conclusion follows 
	immediately. 
	
	For $k=3$ (or any odd integer) there is at each level in the above division exactly 
	one cube $Q$ with its center at the center of $Q_0$, 
	and the statement follows by iteration.
\end{proof}

\subsection{Proof of the Lieb-Thirring inequality}
Now we are able to give a proof of the Lieb-Thirring inequality \eqref{eq:LT_frac}.
\begin{proof}[Proof of Theorem \ref{thm:LT_frac}] 
	By a standard approximation argument we can assume that $\rho_\Psi$ is 
	supported in a finite cube $Q_0 \subset \R^d$. 
	For every $\Lambda\le \int_{\R^d} \rho_\Psi = N$, 
	by applying the Covering Lemma~\ref{lem:covering} 
	with $f=\rho_\Psi$, $k=2$ and $\alpha=2s/d$, 
	we can divide $Q_0$ into disjoint sub-cubes $Q$'s such that 
	$\int_Q \rho_\Psi \le \Lambda$ for all $Q$ and
	\begin{align} \label{eq:LT-partition}
		\sum_Q \frac{1}{|Q|^{2s/d}} \left[\left(\int_Q \rho_\Psi\right)^2 - \frac{\Lambda}{a} \int_Q \rho_\Psi \right] \ge 0,
	\end{align}
	with
	$$
		a:= \frac{2^d}{2}\left(1 + \sqrt{1 + \frac{1-2^{-d}}{2^{d\alpha}-1} } \right).
	$$
	Next, from Lemma~\ref{lem:local-ex}, 
	Lemma~\ref{lem:many-body-local-uncertainty} and \eqref{eq:LT-partition}, 
	it follows that
	\begin{align} \label{eq:LT-final}
		&\left\langle \Psi, \left( \sum_{i=1}^N (-\Delta_{i})^s 
		+ \sum_{1\le i<j \le N} \frac{1}{|x_i-x_j|^{2s}} \right) \Psi \right\rangle \nn\\
		\ge &\sum_Q \left[  \frac{1}{C} \frac{\int_Q \rho_\Psi^{1+2s/d}}{\Big(\int_Q \rho_\Psi \Big)^{2s/d}} 
		- \frac{C}{|Q|^{2s/d}}\int_Q \rho_\Psi + 
		 \frac{1}{2d^{s}|Q|^{2s/d}} \left( \Big(\int_Q \rho_\Psi \Big)^2 - \int_Q \rho_\Psi \right) \right] \nn\\
		\ge &  \frac{1}{C\Lambda^{2s/d}} \int_{\R^d} \rho_\Psi^{1+2s/d} + \left( \frac{\Lambda}{a} - 2d^{s}C-1  \right)
		\sum_Q \frac{1}{2d^{s}|Q|^{2s/d}} \int_Q \rho_\Psi,
	\end{align}
	for every $0<\Lambda\le N$ and for some constant $C>0$ depending only on 
	$d\ge 1$ and $s>0$. Here in the last inequality in \eqref{eq:LT-final} 
	we have used $\int_Q \rho_\Psi \le \Lambda$ for all cubes $Q$'s. 

Finally, using \eqref{eq:LT-final} for $\Lambda=(2d^{s}C+1)a=:\Lambda_0$ 
if $N>\Lambda_0$, and using \eqref{eq:LT-proof-N-small} if $N \le \Lambda_0$, 
we find that 
\begin{align*} 
	\left\langle \Psi, \left( \sum_{i=1}^N (-\Delta_{i})^s 
	+ \sum_{1\le i<j \le N} \frac{1}{|x_i-x_j|^{2s}} \right) \Psi \right\rangle		\ge C \int_{\R^d} \rho_\Psi^{1+2s/d} 
\end{align*}
for a constant $C>0$ depending only on $d$ and $s$. 
The proof is complete.
\end{proof}

\begin{remark}\label{rem:var_coupling}
	Note that, in the case that a coupling parameter $\lambda>0$
	is introduced as in \eqref{eq:LT_frac_coupling},
	a straightforward adaptation of \eqref{eq:LT-final} yields
	$C(\lambda) = C$ for $\lambda \ge 1$
	and $C(\lambda) \sim \lambda^{2s/d}$ for $\lambda < 1$.
\end{remark}

\begin{remark}[Explicit constant] \label{rem:explicit-constant}
	It is possible to derive an explicit constant $C$ in \eqref{eq:LT_frac}. 
	Let us consider for example the case $s=1$ and $d=3$. 
	By the Hoffmann-Ostenhof inequality (see Lemma~\ref{lem:HO}) 
	and Sobolev's inequality,
	\begin{align*}
	\left\langle \Psi, \sum_{i=1}^N -\Delta_{i} \Psi\right\rangle 
		\geq \langle\sqrt{\rho_{\Psi}},(-\Delta)\sqrt{\rho_{\Psi}}\,\rangle 
		\ge C_{\rm S} \left( \int_{\R^3}\rho_{\Psi}^3 \right)^{1/3}
		\ge C_{\rm S} \frac{ \int_{\R^3}\rho_{\Psi}^{5/3}}{\left( \int_{\R^3} \rho_\Psi \right)^{2/3}}.
	\end{align*}
	Moreover, combining the Hoffmann-Ostenhof inequality and the Poincar\'e-Sobolev inequality 
	$$ 
	\norm{\nabla u}_{L^2(Q)}^2 \ge C_{\rm{P}} \norm{u-\frac{1}{|Q|}\int_Q u}_{L^6(Q)}^2 
	$$
	as in \cite{FraSei-12}, we get
	\begin{align*}
		&\left\langle \Psi, \sum_{i=1}^N -\Delta_{i} \Psi\right\rangle 
		\geq \langle\sqrt{\rho_{\Psi}},(-\Delta)\sqrt{\rho_{\Psi}}\,\rangle
		\ge \sum_{Q} \norm{\nabla \sqrt{\rho_{\Psi}}}_{L^2(Q)}^2\\
		&\quad \geq C_{\rm{P}}\sum_{Q}\norm{\sqrt{\rho_{\Psi}}-|Q|^{-1}\int_{Q}\sqrt{\rho_{\Psi}}\,}_{L^6(Q)}^2\\
		&\quad \geq \sum_Q \left[C_{\rm{P}}(1-\eps) \left(\int_Q \rho_{\Psi}^{5/3}\right)\left(\int_Q \rho_{\Psi} \right)^{-2/3}
		- C_{\rm{P}}(\eps^{-1}-1)\frac{1}{|Q|^{2/3}}\int_{Q}\rho_{\Psi}\right]
	\end{align*}
	for any $\eps \in (0,1)$. 
	From these kinetic lower bounds, following the above proof of 
	Theorem \ref{thm:LT_frac}, we find that \eqref{eq:LT_frac} holds true with 
	$$
		C = \frac{\min \{(1-\eps)C_{\rm{P}}, C_{\rm S}\}}{\Lambda_0^{2/3}} , \quad \Lambda_0= a(1+6 C_{\rm{P}}(\eps^{-1}-1)).
	$$
	Here we can take
	$$
		C_{\rm S}= \frac{3}{4} (2\pi^2)^{2/3}, \ 
		C_{\rm P}=\frac{27}{16(1+3^{2/3})^2(2\pi)^{4/3}} 
		\quad \text{and} \quad  
		a=4 + \frac{\sqrt{186}}{3}
	$$
	(the sharp value of $C_{\rm S}$ 
	can be inferred from \cite{Aubin-76,Talenti-76} 
	and the value of $C_{\rm P}$ is obtained by following \cite[Lemma 1]{FraSei-12} 
	but it may not be optimal). 
	Then optimizing over $0<\eps<1$ shows that \eqref{eq:LT_frac} holds true with 
	$$
		C = 0.002384.
	$$
	Although this explicit constant is far from optimal, 
	it is already a significant improvement over \cite{LunPorSol-14}.
\end{remark}

\subsection{Coupling parameter and optimal constant}\label{sec:proof_constants}
Let us here consider the behavior of the optimal constant of \eqref{eq:LT_frac_coupling}
as a function of the coupling parameter $\lambda$,
$$
	C_{\rm BLT}(\lambda) 
	:= \inf_{N \ge 2} \ \inf_{\substack{\Psi \in \cH^s_{d,N} \\ \norm{\Psi}_2=1}}
	\frac{ \left\langle \Psi, \left( \sum_{i=1}^N (-\Delta_{i})^s 
	+ \lambda W_s \right) \Psi \right\rangle
	}{ \int_{\R^d} \rho_{\Psi}^{1+2s/d} },
	\quad \lambda \ge 0,
$$
where $\Psi$ is in the form domain
$$
	\cH^s_{d,N} := \left\{ \Psi \in H^{s}(\R^{dN}) : \int_{\R^{dN}} W_s|\Psi|^2 < \infty \right\}, \ 
	W_s(x) := \!\!\!\!\sum_{1\le i<j \le N}\! \frac{1}{|x_i-x_j|^{2s}}.
$$

Note that the parameter $\lambda$ cannot be removed by scaling and we are interested in the behavior of the optimal constant of \eqref{eq:LT_frac_coupling} in the limits $\lambda \to 0$ and $\lambda \to \infty$. We have

\begin{proposition} \label{prop:constants}
	The optimal constant $C_{\rm BLT}(\lambda)$ is monotone increasing and concave
	as a function of $\lambda$, and satisfies the following:
	\begin{enumerate}[i)]
		\item For all $\lambda >0$, any $d \geq 1$ and all $s > 0$ we have
		\begin{align*}
			0< C_{d,s} \min\{1,\lambda^{2s/d}\} \le C_{\BLT}(\lambda) \leq C_{\GN},
		\end{align*}
		where $C_{d,s}>0$ is a constant independent of $\lambda$ and $C_{\GN}$ 
		is the optimal constant of the one-body fractional Gagliardo-Nirenberg inequality,
		\begin{align}\label{eq:GN_frac_inf}
			C_{\GN} := \inf_{\substack{u \in H^s(\R^d)\\ \norm{u}_2 = 1}} \frac{\langle u, (-\Delta)^s u \rangle}{\int_{\R^d} |u|^{2 (1+2s/d)}}.
		\end{align}
		
		\item We have, for all $d \geq 1$ and any $s>0$,
		\begin{align*}
			\lim_{\lambda \to 0} C_{\BLT} (\lambda) = C_{\rm BLT}(0).
		\end{align*}
		Moreover, for $2s<d$ we have $C_{\BLT}(\lambda) \sim \lambda^{2s/d}$ as $\lambda \to 0$, and in particular $C_{\BLT}(0) = 0$.
	\end{enumerate}
\end{proposition}

In addition, we believe the following to be true:

\begin{conjecture*}
	The optimal constant $C_{\BLT}(\lambda)$ also satisfies:
	\begin{enumerate}[i)]
		\item[iii)]
		$C_{\BLT}(0) > 0$ for $2s > d$.
		
		\item[iv)]
		For all $d \geq 1$ and any $s>0$ we have
		\begin{align*}
			\lim_{\lambda \to \infty} C_{\BLT}(\lambda) = C_{\GN}.
		\end{align*}
	\end{enumerate}
\end{conjecture*}

The proof of Proposition~\ref{prop:constants} will be given below. For $2s<d$, the limit $\lambda \to 0$ 
corresponds to the situation of non-interacting bosons, and by taking the trial wave functions $\Psi=u^{\otimes N}$ 
one can see immediately that $C_{\rm BLT}(\lambda)\to 0$. However, for $2s \ge d$ the situation is more difficult 
because any wave function in $\cH^s_{d,N}$ must vanish  
on the diagonal set 
$$
	\bDelta = \{(x_i)_{i=1}^N \in (\R^d)^N : x_i=x_j~\text{for some}~i\ne j\}
$$
and in particular the trial wave functions $u^{\otimes N}$ are not allowed. 

When $d=s=1$, the operator in \eqref{eq:LT_frac_coupling}
is that of the Calogero-Sutherland model \cite{Cal-69,Suth-71}, 
and the limit $\lambda \to 0$
on the space $L^2_{\sym}$ of symmetric wave functions 
is actually equivalent to non-interacting fermions. In fact, 
$\cH^1_{1,N} \cap L^2_{\sym} = H^1_0(\R^N \setminus \bDelta) \cap L^2_{\sym}$ 
(see \cite[Theorem 2]{LunSol-14})
and it is well known \cite{Girardeau-60} that any such wave function vanishing 
on the diagonal set is equal to an anti-symmetric wave function up to multiplication 
by an appropriate sign function. Therefore, $C_{\BLT}(0)$ is exactly the optimal 
constant $C_{\LT}$ of the fermionic Lieb-Thirring inequality \eqref{eq:LT-fermion}, 
which is {\em conjectured} \cite{LieThi-76} to be $C_{\GN} = \pi^2/4$.

When $d=1$ and $s=2$, the condition of anti-symmetry is however not strong enough to 
ensure that the wave function is in the quadratic form domain $\cH^2_{1,N}$,
which can be seen readily by taking the two-body 
state $\Psi(x_1,x_2) = C(x_1-x_2)e^{-|x_1|^2-|x_2|^2} \notin \cH^2_{1,2}$. 
In this case we expect $C_{\BLT}(0) > C_{\LT}$ because of the more restricted domain.

For $d \ge 2$ the situation is yet more difficult: Because of the
connectedness of the configuration space $(\R^d)^N \setminus \bDelta$ 
there is no simple boson-fermion correspondence 
for functions vanishing on $\bDelta$, for any $s>0$.
Furthermore, if $s-d/2 \in \{0,1,2,...\}$, then the interaction operator 
$W_s$ cannot be controlled by the kinetic operator $\sum_i (-\Delta_i)^s$ 
by means of the Hardy inequality (see \cite{Solom-94,Yafaev-99}), which makes it
difficult to compare $\cH^s_{d,N}$ with $H^s_0(\R^{dN} \setminus \bDelta)$. 
It is an interesting open question to determine the complete behavior of
$C_{\BLT}(0)$ in the general case $2s \ge d$. 
We expect $C_{\BLT}(0) > 0$ for $2s>d$ because in this case 
$H^s_0(\R^{dN} \setminus \bDelta) \neq H^s(\R^{dN})$
(by Sobolev embedding),
and a smooth vanishing condition for $\Psi$ on $\bDelta$ 
should imply a non-trivial local exclusion principle.
In the critical case $2s=d$ it may happen that $C_{\BLT}(0)=0$,
as can be seen for $d=2$, $s=1$ using the ground state of 
a gas of hard disks in a dilute limit \cite{LieYng-01}.

On the other hand,
in the limit $\lambda \to \infty$ of strong interaction, we expect the inter-particle distance to go to infinity, 
and hence the optimal constant should tend to the one-body 
constant $C_{\GN}$ of \eqref{eq:GN_frac_inf}. It seems that proving this 
would require a concentration-compactness method for many-body systems 
which is not available to us at the moment. We also notice that in the physically 
most interesting case $d=3$ and $s=1$, the {\em conjectured} optimal constant 
in the fermionic Lieb-Thirring inequality \eqref{eq:LT-fermion} \cite{LieThi-76} is strictly smaller than  $C_{\GN}$. 

\begin{proof}[Proof of Proposition~\ref{prop:constants}]	
	We first note that $\lambda \mapsto C_{\BLT}(\lambda)$ is the infimum of
	monotone increasing affine functions
	(denoting $\hat{T} := \sum_{i=1}^N (-\Delta_{i})^s$)
	$$
		\lambda \mapsto
		\frac{ \langle \Psi, \hat{T} \Psi \rangle }
		{ \int_{\R^d} \rho_{\Psi}^{1+2s/d} }
		+ \lambda\frac{ \langle \Psi, W_s \Psi \rangle }
		{ \int_{\R^d} \rho_{\Psi}^{1+2s/d} },
	$$
	and hence monotone increasing and concave.
	
	\noindent
	{\bf Proof of (i).}
	From Remark~\ref{rem:var_coupling} we obviously have
	$$
		C_{\BLT}(\lambda) \geq C_{d,s}\min\{1,\lambda^{2s/d}\} > 0,
	$$ 
	so it remains to prove that $C_{\BLT}(\lambda) \le C_{\GN}$.
	Following \cite[Theorem 19]{LunSol-14}, we take a sequence of trial states 
	$$
		\Psi_{N,R}(x) 
		:= \frac{1}{\sqrt{N!}} \sum_{\sigma\in S_N} u_{\sigma(1)}(x_1)\cdots u_{\sigma(N)}(x_N)
		\ \in \cH^s_{d,N} \cap L^2_{\sym},
	$$
	with
	$$
		u_i(x) := u^R(x - Ry_i),
	$$
	where $u^R \in C^\infty_0(B(0,R/3))$ is a minimizing sequence of 
	$L^2$-normalized functions for \eqref{eq:GN_frac_inf}
	(s.t. both numerator and denominator remain finite),
	and $y_i$ are $N$ disjoint points in $\R^d$, 
	with $|y_i-y_j| > 1$ for $i \neq j$.
	Since the supports of the $u_i$'s are disjoint, 
	one readily computes that 
	\begin{align}\label{eq:uni_lambda_trial}
		C_{\rm BLT}(\lambda) 
		\leq \frac{N \left\langle u^{R}, (-\Delta)^s u^{R}\right\rangle + \lambda C N^2 R^{-2s}}
		{N \int_{\R^d}|u^R|^{2(1+2s/d)}},
	\end{align}
	and the right hand side of \eqref{eq:uni_lambda_trial} converges to 
	$C_{\GN}$ in the limit $R \to \infty$.
	Note that we could also have taken $\Psi_{N,R}$ as an anti-symmetric state (a Slater determinant). 

	\noindent
	{\bf Proof of (ii).} 
	We will first show that for any $d \geq 0$ and all $s>0$, 
	$\lim_{\lambda \to 0} C_{\BLT}(\lambda) = C_{\BLT}(0)$, with
	\begin{align}\label{eq:C_0_def}
		C_{\BLT}(0)
		:= \inf_{N \ge 2} \ \inf_{\substack{\Psi \in \cH^s_{d,N} \\ \norm{\Psi}_2 = 1}}
		\frac{ \langle \Psi, \hat{T} \Psi \rangle }{ \int_{\R^d} \rho_{\Psi}^{1+2s/d} }.
	\end{align}
	To do so, we first pick a minimizing sequence $(\Psi_{N_k})_{k \in \N}$ 
	for \eqref{eq:C_0_def} 
	(with each $\Psi_{N_k} \in \cH^s_{d,N_k}$ and normalized). 
	Next, we have
	\begin{align}\label{eq:C_0_conv}
		0 \le C_{\BLT}(\lambda) - C_{\BLT}(0) 
		&\leq \frac{ \left\langle \Psi_{N_k}, (\hat{T} + \lambda W_s) \Psi_{N_k}\right\rangle }
		{ \int_{\R^d} \rho_{\Psi_{N_k}}^{1+2s/d} } - C_{\BLT}(0)\nn\\
		&=\lambda \frac{\left\langle \Psi_{N_k}, W_s \Psi_{N_k}\right\rangle}{ \int_{\R^d} \rho_{\Psi_{N_k}}^{1+2s/d} } 
		+ \frac{ \langle \Psi_{N_k}, \hat{T} \Psi_{N_k} \rangle }{ \int_{\R^d} \rho_{\Psi_{N_k}}^{1+2s/d} } - C_{\BLT}(0).
	\end{align}
	Given any $\eps > 0$, the last term of \eqref{eq:C_0_conv} is clearly less 
	than $\eps$ for $k \in \N$ sufficiently large,
	while the first term remains bounded. 
	With such $k$ fixed, we then choose 
	$\lambda < \eps (\int_{\R^d} \rho_{\Psi_{N_k}}^{1+2s/d}) / \left\langle \Psi_{N_k}, W_s \Psi_{N_k}\right\rangle$,
	so that $C_{\BLT}(\lambda) - C_{\BLT}(0) < 2\eps$.
	
	In the case $2s<d$ we have $C_{\BLT}(\lambda) \sim \lambda^{2s/d}$ as $\lambda \to 0$, which can be seen by taking a 
	bosonic trial state $\Psi = u^{\otimes N} \in \cH_{d,N}^s$ and letting $N \sim \lambda^{-1}$. 
\end{proof}

\subsection{A note about fermions and weaker exclusion} \label{sec:fermion}
In this subsection we explain how to adapt our above proof of Theorem~\ref{thm:LT_frac}
to show the fermionic inequality \eqref{eq:LT_frac_fermions}
\begin{align*}
	\left\langle \Psi, \sum_{i=1}^N (-\Delta_{i})^s \Psi \right\rangle 
	\ge C \int_{\R^d} \rho_{\Psi}(x)^{1+2s/d} \,dx
\end{align*}
for all $d\ge 1$ and $s>0$, where the wave function $\Psi$ satisfies the 
anti-symmetry \eqref{eq:anti-symmetry}. 
In this case the kinetic energy 
not only contributes to a local uncertainty principle
as in Lemma~\ref{lem:many-body-local-uncertainty}, 
but also to a local exclusion
principle of the following weaker form:

\begin{lemma}[Local exclusion for fermions] \label{lem:local-exclusion-fermions} 
	For any $d\geq 1$, $s > 0$ 
	there is a constant $C>0$ depending only on $d$ and $s$ such that for all $N \in \N$,
	for every $L^2$-normalized function $\Psi\in H^s(\R^{dN})$ 
	satisfying the anti-symmetry \eqref{eq:anti-symmetry}, and 
	for an arbitrary collection of disjoint cubes $Q$'s in $\R^d$, 
	\begin{align}\label{eq:local-exclusion-fermions}
		\left\langle \Psi , \sum_{i=1}^N (-\Delta_i)^s \Psi \right \rangle
		\ge \sum_Q \frac{C}{|Q|^{2s/d}} \left[ \int_Q \rho_\Psi(x)\,dx - q \right]_+,
	\end{align}
	where
	$q := \# \{ \textup{multi-indices}\ \alpha : 0 \le |\alpha|<s \}$.
\end{lemma}
\begin{proof}
	First, consider one-body functions
	$u \in H^s(Q)$ where $s=m+\sigma$, $m \in \N$, 
	$\sigma \in [0,1)$.
	In the case that $0<\sigma<1$, we have the 
	fractional Poincar\'e inequality
	(see, e.g., \cite[Lemma 2.2]{HurVah-13})
	$$
		\|u\|_{\dot H^s(Q)}^2
		\ge \frac{C}{|Q|^{2\sigma/d}} \sum_{|\alpha|=m} \left\|D^\alpha u 
			- \frac{1}{|Q|}\int_{Q} D^\alpha u\right\|_{L^2(Q)}^2,
	$$
	while for $|\alpha|=m$ we have 
	(by iteration of Poincar\'e's inequality) 
	$$
		\|D^\alpha u\|_{L^2(Q)}^2
		\ge \frac{C}{|Q|^{2m/d}} \|u\|_{L^2(Q)}^2,
		\quad \text{if $\int_Q D^\beta u = 0$ for all $0 \le |\beta|<m$}. 
	$$
	Note that 
	$\int_Q D^\alpha u = \langle 1, T_\alpha u \rangle 
		= \langle T_\alpha^* 1, u\rangle$,
	where the operator $u \mapsto T_\alpha(u) := D^\alpha u$,
	$|\alpha|\le m$,
	is relatively bounded w.r.t. the form domain $H^s(Q)$.
	Hence we can treat these orthogonality conditions by considering
	the $q$-dimensional subspace 
	$\V_s := \Span \{ T_\alpha^* 1 : 0 \le |\alpha| < s \}$.
	On $H^s(Q) \cap \V_s^\perp$ we then have
	$$
		\|u\|_{\dot H^s(Q)}^2 \ge \frac{C}{|Q|^{2s/d}} \|u\|_{L^2(Q)}^2,
	$$
	and in general, by taking out the projection onto $\V_s$,
	$$
		(-\Delta)^s|_{H^s(Q)} \ge \frac{C}{|Q|^{2s/d}}(\1 - P_{\V_s}).
	$$
	
	Now we proceed as in Lemma \ref{lem:many-body-local-uncertainty}, 
	although because of the anti-symmetry of $\Psi$, 
	the one-body functions $u_n$ all have norm less than unity
	(again, see e.g. \cite{LieSei-10}).
	We then obtain
	$$
		\left\langle \Psi, \sum_{i=1}^N (-\Delta_{i})^s \Psi \right\rangle 
		\ge \sum_{n\ge 1} \sum_Q \norm{u_n}_{\dot{H}^s(Q)}^2
		\ge \sum_Q \frac{C}{|Q|^{2s/d}}\left[ \sum_{n\ge 1} \|u_n\|_{L^2(Q)}^2 - q \right]_+,
	$$
	which proves the lemma.
\end{proof}

We note that the Covering Lemma~\ref{lem:covering} can be also adapted to 
apply to the weaker form of the exclusion principle.
This could be useful not only for fermions but also in situations when 
other types of interactions are present
(cf. \cite{LunSol-13,FraSei-12,LunSol-14,LunPorSol-14}).

\begin{lemma}[Covering lemma with weaker exclusion] \label{lem:covering-weak} 
	Let $Q_0$ be a cube in $\R^d$ and let $0\le f\in L^1(Q_0)$ satisfy 
	$\int_{Q_0} f \ge \Lambda>0$. Then $Q_0$ can be divided into disjoint 
	sub-cubes $Q$'s such that
	\begin{itemize}
		\item For all $Q$,
		$$
			\int_{Q} f < \Lambda.
		$$
		
		\item For all $\alpha>0$, $q \ge 0$ and integer $k \ge 2$, 
		\bq \label{eq:bound-f*f-f-weak}
			\sum_{Q} \frac{1}{|Q|^{\alpha}} \left( \left[\int_{Q} f - q \right]_+ 
				- b \int_{Q} f \right) \ge 0,
		\eq
		where 
		$$
			b:= \left( 1- \frac{qk^d}{\Lambda} \right) \frac{k^{d\alpha}-1}{k^{d\alpha} + k^d - 2}.
		$$
		
		\item If $k=3$, then the center of $Q_0$ coincides with exactly one sub-cube $Q$,
		and the distance from every other sub-cube $Q$ to the center of $Q_0$
		is not smaller than $|Q|^{1/d}/2$.
	\end{itemize}
\end{lemma}
\begin{proof}
	We proceed with the same division procedure as in the proof of
	Lemma~\ref{lem:covering}.
	Instead of \eqref{eq:Fi-smallest} we have
	\begin{equation} \label{eq:Fi-smallest-weak}
		\sum_{Q\in \cF_i, |Q|=m_i} \frac{1}{|Q|^\alpha} 
			\left( \left[ \int_Q f -q \right]_+ - b \int_Q f \right) 
		\ge \frac{1}{m_i^\alpha} \left( (1-b)\Lambda - qk^d \right),
	\end{equation}
	and instead of \eqref{eq:Fi-larger} we have
	\begin{multline} \label{eq:Fi-larger-weak}
		\sum_{Q\in \cF_i, |Q|>m_i} \frac{1}{|Q|^\alpha} \left( 
			\left[ \int_Q f -q \right]_+ - b \int_Q f \right) 
		\ge - b\Lambda \sum_{Q\in \cF_i, |Q|>m_i} \frac{1}{|Q|^\alpha}\\
		\ge - b\Lambda \sum_{j \ge 1} \frac{k^d-1}{ (k^{dj}m_i)^\alpha} 
		= - \frac{b\Lambda}{m_i^\alpha} \frac{k^d-1}{k^{d\alpha}-1}.
	\end{multline}
	Hence,
	\begin{align*} 
		\sum_{Q\in \cF_i} \frac{1}{|Q|^\alpha} \left( 
			\left[ \int_Q f -q \right]_+ - b \int_Q f \right)
		\ge \frac{1}{m_i^\alpha} \left( 
			\Lambda - qk^d - b\Lambda\left(1 + \frac{k^d-1}{k^{d\alpha}-1}\right)
			\right),
	\end{align*} 
	from which the lemma follows.
\end{proof}

From the local uncertainty in Lemma~\ref{lem:many-body-local-uncertainty}, 
the local exclusion in Lemma \ref{lem:local-exclusion-fermions} and the 
Covering Lemma \ref{lem:covering-weak}, one can prove the fermionic 
Lieb-Thirring inequality \eqref{eq:LT_frac_fermions} by proceeding similarly 
as in the proof of Theorem \ref{thm:LT_frac}. 
The details are left to the reader.

\begin{remark}
	From Lemma \ref{lem:local-ex} and the elementary inequality 
	$(a^2-a)_+ \ge (a-1)_+$, $a\ge 0$, we obtain  
	the following analogue of \eqref{eq:local-exclusion-fermions}
	for pair-interactions:
	\begin{align} \label{eq:exclusion-DL-1}
		\left\langle \Psi , \sum_{1\le i<j \le N} \frac{1}{|x_i-x_j|^{2s}} \Psi \right \rangle 
		\ge \sum_Q\frac{1}{2d^s|Q|^{2s/d}} \left[ \int_Q \rho_\Psi - 1 \right]_+
	\end{align} 
	for every normalized function $\Psi\in L^2(\R^{dN})$. In our proofs of the Lieb-Thirring 
	inequality \eqref{eq:LT_frac} and the Hardy-Lieb-Thirring inequality \eqref{eq:HLT_frac} presented later, 
	we can certainly use \eqref{eq:exclusion-DL-1} instead of \eqref{eq:local-exclusion}  
	(we then obtain similar inequalities but with worse constants).     
\end{remark}

\section{Hardy-Lieb-Thirring inequality} \label{sec:HLT}
In this section we prove Theorem~\ref{thm:HLT_frac}. 
We will need to strengthen the local uncertainty principle in 
Section~\ref{sec:LT} to account for the Hardy term,
and to do this we also need a localization method for fractional kinetic energy.

\subsection{Local uncertainty for centered cubes}
The following local uncertainty principle is crucial for our proof. 
\begin{lemma}[Local uncertainty for centered cubes] \label{lem:HLT-uncertainty} 
	For every cube $Q\subset \R^d$ centered at $0$, we have 
	\begin{align} \label{eq:HLT-uncertainty}
		\|u\|_{\dot H^s(Q)}^2 - \mathcal{C}_{d,s}\int_Q \frac{|u(x)|^2}{|x|^{2s}}\,dx 
		\ge \frac{1}{C}\frac{\int_Q |u|^{2(1+2s/d)}}{\Big(\int_Q |u|^2\Big)^{2s/d}} - \frac{C}{|Q|^{2s/d}} \int_Q |u|^2   
	\end{align}
	for a constant $C>0$ depending only on $d\ge 1$ and $s>0$. 
\end{lemma}

Note that this local uncertainty principle is significantly stronger than 
the one in Lemma \ref{lem:local_uncert} because the left side of 
\eqref{eq:HLT-uncertainty} can even be negative. 
Our strategy is to replace $u$ by $\chi u$ where $\chi$ is a smooth function 
supported in a neighborhood of the origin, and then apply the Hardy 
inequality with remainder term for $\chi u \in H^s(\R^d)$. 
To implement the localization procedure, we also need the following lemma
which controls the error terms.

\begin{lemma}[A fractional IMS localization formula] \label{lem:fractional-IMS} 
	Let $\Omega$ be a bounded open domain in $\R^d$ with $d\ge 1$. 
	Let $\chi,\eta:\R^d\to [0,1]$ be two smooth functions such that 
	$\chi(x)^2+\eta(x)^2\equiv 1$ and $\chi$ is supported in 
	a compact subset of $\Omega$. 
	Then for every $s>0$, there exists $t\in [0,s)$ and a constant $C>0$ 
	such that for every $u\in H^s(\Omega)$,
	\begin{align}  \label{eq:fractional-IMS}
		\left| \|u\|_{\dot H^s(\Omega)}^2- \|\chi u\|_{\dot H^s(\Omega)}^2 
			- \|\eta u\|_{\dot H^s(\Omega)}^2 \right| 
		\le C \left( \|\chi u\|_{H^{t}(\Omega)}^2 
			+ \|\eta u\|_{ H^{t}(\Omega)}^2 \right).
	\end{align}
\end{lemma}
\begin{remark} 
	It will be clear from the proof of Lemma \ref{lem:fractional-IMS} 
	(provided below) that if $s\in \N$ then $t=s-1$, 
	and if $s=m+\sigma$ with $m\in \{0,1,2,\dots\}$ and $0<\sigma<1$
	then we can take $t=s-\eps$ for any $0<\eps<\min\{\sigma, 1-\sigma\}$.
\end{remark}

Note that such a localization bound is well known when $0<s\le 1$.
In the simplest case $s=1$, thanks to the IMS formula 
(cf. \cite[Theorem 3.2]{CycFroKirSim-87})
\begin{align*}
	|\nabla u|^2 = |\nabla(\chi u)|^2 + |\nabla (\eta u)|^2 - (|\nabla \chi|^2 + |\nabla \eta|^2)|u|^2,
\end{align*}
we obtain the estimate \eqref{eq:fractional-IMS} (with $t=0$) immediately:
\begin{align*} 
	\left| \|u\|^2_{\dot H^1(\Omega)} - \|\chi u\|^2_{\dot H^1(\Omega)} - \|\eta u\|^2_{\dot H^1(\Omega)} \right|
	=\int_\Omega(|\nabla \chi|^2 + |\nabla \eta|^2)|u|^2 \le  C \int_\Omega |u|^2.
\end{align*}
When $0<s<1$, the estimate
\begin{align*} 
	\left| \|u\|_{\dot H^s(\Omega)}^2- \|\chi u\|_{\dot H^s(\Omega)}^2 - \|\eta u\|_{\dot H^s(\Omega)}^2 \right| 
	\le C \int_\Omega |u|^2
\end{align*}
follows from the representation \eqref{eq:fractional-form}
$$
	\|u\|_{\dot H^s(\Omega)}^2=c_{d,s} \iint_{\Omega\times \Omega} \frac{|u(x)-u(y)|^2}{|x-y|^{d+2s}}\,dxdy
$$
and the elementary identity (which goes back to a suggestion of Michael Loss 
and was used in \cite{LieYau-88})
\begin{multline}\label{eq:Loss-formula}
	|\chi(x)u(x)-\chi(y)u(y)|^2 + |\eta(x)u(x)-\eta(y)u(y)|^2 -|u(x)-u(y)|^2 \\
	= \left[(\chi(x)-\chi(y))^2 + (\eta(x)-\eta(y))^2 \right] \Re[\overline{u(x)}u(y)].
\end{multline}
However, the proof of \eqref{eq:fractional-IMS} for  $s>1$ is rather involved 
and we defer it to the next subsection. 
In the following, we will give a proof of Lemma~\ref{lem:HLT-uncertainty} 
using Lemma~\ref{lem:fractional-IMS}. 

\begin{proof}[Proof of Lemma \ref{lem:HLT-uncertainty}]
	Since the inequality \eqref{eq:HLT-uncertainty} 
	that we wish to prove is dilation invariant, 
	we can assume without loss of generality that $|Q|=1$. 
	Let $\chi,\eta:\R^d \to [0,1]$ be two smooth functions such that 
	$\chi^2(x)+\eta^2(x)\equiv 1$, $\chi(x)\equiv 1$ when $|x|\le 1/4$ 
	and $\chi(x)\equiv 0$ when $|x|\ge 1/3$. 
	By using $\eta^2 |u|^2/|x|^{2s}\le 3^{2s} \eta^2 |u|^2$ and 
	Lemma~\ref{lem:fractional-IMS} we obtain
	for some $t \in [0,s)$
	\begin{align} \label{eq:HLT-proof-1}
		\|u\|_{\dot H^s(Q)}^2 - \mathcal{C}_{d,s} \int_Q\frac{|u|^2}{|x|^{2s}}\,dx
		&\ge \|\chi u \|_{\dot H^s(Q)}^2  - \mathcal{C}_{d,s} \int_Q\frac{|\chi u|^2}{|x|^{2s}}\,dx \nn\\
		&\quad + \|\eta u \|_{\dot H^s(Q)}^2 -  C_1\|\chi u \|_{H^{t}(Q)}^2 - C_1 \|\eta u\|_{H^{t}(Q)}^2
	\end{align}
	for some constant $C_1>0$ depending only on $d,s,t$ (and $\chi$).
	
	Since $\chi$ has compact support, $\chi u$ can be considered as a function in 
	$H^s(\R^d)$. Therefore, by the Gagliardo-Nirenberg type inequality 
	\eqref{eq:Hardy-improved-GN} (there taking $t = s/(1+2s/d)$),
	\begin{align} \label{eq:HLT-proof-2}
		\frac{1}{2} \left( \|\chi u \|_{\dot H^s(\R^d)}^2 - \mathcal{C}_{d,s} \int_{\R^d} \frac{|\chi u|^2}{|x|^{2s}}\,dx \right)
		\ge \frac{1}{C} \frac{ \int |\chi u|^{2(1+2s/d)}}{\Big( \int |\chi u|^2\Big)^{2s/d}}. 
	\end{align}
	Moreover, by using the improved Hardy inequality \eqref{eq:Hardy-improved} 
	and the norm-equivalence \eqref{eq:equiv-norm-Sobolev}, we find
	\begin{multline*}
		\left( \|\chi u \|_{\dot H^s(\R^d)}^2  
		- \mathcal{C}_{d,s} \int_{\R^d} \frac{|\chi u|^2}{|x|^{2s}}\,dx \right)^{t/s} \|\chi u \|_{L^2(\R^d)}^{2(1-t/s)} \\
		\ge  \frac{1}{C} \| \chi u \|_{\dot H^{t}(\R^d)}^2 \ge \frac{1}{C} \|\chi u\|^2_{H^t(\R^d)} - C \|\chi u \|_{L^2(\R^d)}^2,
	\end{multline*}
	which by Young's inequality implies that
	\begin{align} \label{eq:HLT-proof-3}
		\frac{1}{2}\left( \|\chi u \|_{\dot H^s(\R^d)}^2 - \mathcal{C}_{d,s} \int_{\R^d} \frac{|\chi u|^2}{|x|^{2s}}\,dx \right) 
		\ge C_1 \| \chi u \|_{H^{t}(\R^d)}^2 - C \|\chi u \|_{L^2(\R^d)}^2,
	\end{align}
	with $C_1$ as in \eqref{eq:HLT-proof-1} and
	a (large) constant $C>0$ depending only on $d,s,t$.  
	
	For the function $\eta u$, by the local uncertainty in Lemma~\ref{lem:local_uncert},
	\begin{align}  \label{eq:HLT-proof-4}
		\frac{1}{2} \|\eta u \|_{\dot H^s(Q)}^2 
		\ge \frac{1}{C} \frac{ \int_Q |\eta u|^{2(1+2s/d)}}{\Big( \int_Q |\eta u|^2\Big)^{2s/d}} -  C \|\eta u\|^2_{L^2(Q)}.
	\end{align}
	By using the extension and interpolation arguments as in the proof of 
	Lemma~\ref{lem:local_uncert}, we obtain
	\begin{align*}
	\| \eta u \|_{H^{s}(Q)}^{t/s} \|\eta u\|^{1-t/s}_{L^2(Q)} \ge C \|\eta u\|_{H^t(Q)},
	\end{align*}
	which, together with the norm-equivalence \eqref{eq:equiv-norm-Sobolev}, 
	gives the estimate
	\begin{align} \label{eq:HLT-proof-5}
		\frac{1}{2}\| \eta u \|_{\dot H^{s}(Q)}^2 \ge C_1 \|\eta u\|_{H^t(Q)}^2 - C \|\eta u\|_{L^2(Q)}^2
	\end{align}
	for a (large) constant $C>0$ depending only on $d,s,t$.  
	
	By summing inequalities 
	\eqref{eq:HLT-proof-1}-\eqref{eq:HLT-proof-2}-\eqref{eq:HLT-proof-3}-\eqref{eq:HLT-proof-4}-\eqref{eq:HLT-proof-5},
	using
	$$
		\|\chi u\|_{L^2(Q)}^2 + \|\eta u \|_{L^2(Q)}^2 = \|u\|_{L^2(Q)}^2
	$$
	and estimating the denominators, we arrive at 
	\begin{align*}
		\|u\|_{\dot H^s(Q)}^2 - \mathcal{C}_{d,s} \int_Q\frac{|u|^2}{|x|^{2s}}\,dx 
		\ge \frac{1}{C} \frac{\int_Q \Big( |\chi u|^{2(1+2s/d)} + |\eta u|^{2(1+2s/d)}\Big) }{\Big(\int_Q |u|^2 \Big)^{2s/d}} - C \|u\|_{L^2(Q)}^2 
	\end{align*}
	for a (large) constant $C>0$ depending only on $d,s$. 
	The final conclusion then follows from the elementary inequality
	$$
		\chi^{2p}+\eta^{2p} \ge 2 \left( \frac{\chi^2+\eta^2}{2}\right)^p = 2^{1-p}, \quad p=1+\frac{2s}{d}>1.
	$$
\end{proof}

\subsection{Proof of the fractional IMS localization formula}
\begin{proof}[Proof of Lemma \ref{lem:fractional-IMS}] 
	{\bf Step 1.} We start with the case $s=m\in \N$. 
	Recall that in our conventions
	$$
		\|u\|_{\dot H^m(\Omega)}^2 = \sum_{|\alpha|=m} \frac{m!}{\alpha!} \int_\Omega |D^\alpha u|^2.
	$$
	Let us consider an arbitrary multi-index $\alpha$ with $|\alpha|=m$. Using
	\begin{align} \label{eq:D(chiu)}
		D^\alpha (\chi u) 
		= \chi D^\alpha u + \sum_{\beta< \alpha} \frac{\alpha!}{\beta!(\alpha-\beta)!}D^{\alpha-\beta} \chi D^{\beta} u
	\end{align}
	and a similar formula for $D^\alpha(\eta u)$, we find that
	\begin{align} \label{eq:D(chiu)-square}
		&|D^\alpha (\chi u)|^2 + |D^\alpha (\eta u)|^2 = (\chi^2 +\eta^2) |D^\alpha u|^2 \nn\\
		&\quad + \left| \sum_{\beta<\alpha} \frac{\alpha!}{\beta!(\alpha-\beta)!}  D^{\alpha-\beta} \chi D^\beta u \right|^2 
		+ \left| \sum_{\beta<\alpha} \frac{\alpha!}{\beta!(\alpha-\beta)!} D^{\alpha-\beta} \eta  D^\beta u \right|^2 \nn\\
		&\quad + 2 \Re \sum_{\beta<\alpha} \frac{\alpha!}{\beta!(\alpha-\beta)!} (\chi D^{\alpha-\beta} \chi 
		+ \eta D^{\alpha-\beta} \eta )  D^{\alpha}\overline{u} D^\beta u.
	\end{align}
	Here, for two multi-indices $\alpha=(\alpha_1,\dots,\alpha_d)$ and 
	$\beta=(\beta_1,\dots,\beta_d)$, the notation $\beta<\alpha$ means 
	$\beta\le \alpha$, namely $\beta_j\le \alpha_j$ for all $1\le j \le d$, 
	and $\beta\ne \alpha$. 
	The first term of the right side of \eqref{eq:D(chiu)-square} is nothing but 
	$|D^\alpha u|^2$ since $\chi^2+\eta^2=1$. 
	The next two terms can be bounded using the Cauchy-Schwarz inequality
	\begin{align*}
		\left| \sum_{\beta<\alpha} \frac{\alpha!}{\beta!(\alpha-\beta)!} D^{\alpha-\beta} \chi D^\beta u \right|^2 
	 	+ \left| \sum_{\beta<\alpha} \frac{\alpha!}{\beta!(\alpha-\beta)!} D^{\alpha-\beta} \eta  D^\beta u \right|^2 
		\le C\sum_{\beta<\alpha}|D^\beta u|^2.
	\end{align*}
	Therefore, by integrating \eqref{eq:D(chiu)-square} and using the triangle 
	inequality we get 
	\begin{align} \label{eq:D(chiu)-square-1}
		&\left| \|\chi u\|_{\dot H^m(\Omega)}^2 + \|\eta u\|_{ \dot H^m(\Omega)}^2 - \|u\|_{\dot H^m(\Omega)}^2 \right| 
		\le C \|u\|_{H^{m-1}(\Omega)}^2 \nn\\
		&\quad + 2 \sum_{|\alpha|=m} \sum_{\beta<\alpha} \frac{\alpha!}{\beta!(\alpha-\beta)!} 
		\left| \iint_{\Omega\times \Omega} (\chi D^{\alpha-\beta} \chi + \eta D^{\alpha-\beta} \eta ) 
		 D^{\alpha}\overline{u} D^\beta u \right|.
	\end{align}
	Now we estimate the last term of \eqref{eq:D(chiu)-square-1}. 
	For every $\alpha$ with $|\alpha|=m$, we can find $0\le \alpha'<\alpha$ 
	and $1\le j\le d$ such that $D^\alpha=\partial_j D^{\alpha'}$. 
	Note that $\chi D^{\alpha-\beta} \chi + \eta D^{\alpha-\beta} \eta$ has support 
	in a compact subset of $\Omega$, so by using integration by parts with respect 
	to the $j$-th coordinate we find that
	\begin{align*}
		&\int_\Omega (\chi D^{\alpha-\beta} \chi + \eta D^{\alpha-\beta} \eta )  D^{\alpha}\overline{u} D^\beta u\\
		&\quad =- \int_\Omega D^{\alpha'} \overline{u} \partial_j \left( (\chi D^{\alpha-\beta} \chi + \eta D^{\alpha-\beta} \eta )  
		D^\beta u \right)\\
		&\quad =- \int_\Omega D^{\alpha'} \overline{u} \left(  \partial_j (\chi D^{\alpha-\beta} \chi + \eta D^{\alpha-\beta} \eta )  
		D^\beta u + (\chi D^{\alpha-\beta} \chi + \eta D^{\alpha-\beta} \eta ) \partial_j D^\beta u \right). 
	\end{align*}
	Therefore, when $|\beta|\le m-2$, by the Cauchy-Schwarz inequality 
	we can estimate
	\begin{align*}
		\left| \int_\Omega  (\chi D^{\alpha-\beta} \chi + \eta D^{\alpha-\beta} \eta )  D^{\alpha}\overline{u} D^\beta u  \right|  
		\le C\|u\|_{H^{m-1}(\Omega)}^2.
	\end{align*}
	On the other hand, if $\beta<\alpha$ and $|\beta|=m-1=|\alpha|-1$, 
	then $D^{\alpha-\beta}=\partial_k$ for some $1\le k \le d$, and hence 
	$$
		\chi D^{\alpha-\beta} \chi+\eta D^{\alpha-\beta} \eta= \frac{1}{2} \partial_{k} \left( \chi^2+\eta^2  \right)=0.
	$$
	Summarizing, \eqref{eq:D(chiu)-square-1} can be simplified to  
	\begin{align}\label{eq:D(chiu)-square-2}
		\left| \|\chi u\|_{\dot H^m(\Omega)}^2 + \|\eta u\|_{\dot H^m(\Omega)}^2 - \|u\|_{\dot H^m(\Omega)}^2 \right| 
		\le  C \|u\|_{H^{m-1}(\Omega)}^2. 
	\end{align}
	Since $\|u\|_{H^{m-1}(\Omega)}^2 \asymp \sum_{0\le n\le m-1} \|u\|_{\dot H^n(\Omega)}^2$, 
	we can continue estimating the right side of \eqref{eq:D(chiu)-square-2} 
	by induction and finally arrive at
	\begin{multline*}
		\left| \|\chi u\|_{\dot H^m(\Omega)}^2 + \|\eta u\|_{ \dot H^m(\Omega)}^2 - \|u\|_{\dot H^m(\Omega)}^2 \right|
		\le C \left( \|\chi u\|_{H^{m-1}(\Omega)}^2 + \|\eta u\|_{H^{m-1}(\Omega)}^2 \right). 
	\end{multline*}
	This ends the proof when $s=m\in \N$.
	
	\noindent
	{\bf Step 2.} Now we consider the case when $s=m+\sigma$ with $m\in \N$ 
	and $0<\sigma<1$. Let us start by considering
	\begin{align*}
		&\|\chi u\|^2_{\dot H^s(\Omega)} 
		= c_{d,\sigma}\sum\limits_{|\alpha|=m} \dfrac{m!}{\alpha!} \iint_{\Omega\times \Omega} \dfrac{|D^\alpha (\chi u) (x) 
		-D^\alpha (\chi u) (y)|^2}{|x-y|^{d+2\sigma}}\,dxdy.
	\end{align*}
	We will always denote by $\alpha$ an arbitrary multi-index with $|\alpha|=m$. 
	Using \eqref{eq:D(chiu)} and the identity 
	$|a+b|^2=|a|^2 + 2\Re (\overline{(a+b)}b)-|b|^2$ 
	(with complex numbers $a$ and $b$), we have 
	\begin{align} \label{eq:HLT-IMS-s>1-a}
		&|D^\alpha(\chi u)(x)-D^\alpha(\chi u)(y)|^2 \nn\\
		&\quad = \Big| \chi(x)D^\alpha u(x) -\chi(y) D^\alpha u(y) \nn\\
		&\qquad + \sum_{\beta<\alpha} \frac{\alpha!}{\beta!(\alpha-\beta)!} 
		\left( D^{\alpha-\beta} \chi(x) D^{\beta} u (x)- D^{\alpha-\beta} \chi(y) D^{\beta} u(y)\right) \Big|^2\nn\\
		&\quad =|\chi(x) D^\alpha u (x) - \chi(y) D^\alpha u (y)|^2  \nn\\
		&\qquad - \Big| \sum_{\beta<\alpha} \frac{\alpha!}{\beta!(\alpha-\beta)!} \left( D^{\alpha-\beta} \chi(x) D^{\beta} 
		u (x)- D^{\alpha-\beta} \chi(y) D^{\beta} u(y)\right) \Big|^2  \nn \\
		&\qquad +  2\Re \sum_{\beta< \alpha} \frac{\alpha!}{\beta!(\alpha-\beta)!} \Big(D^\alpha (\chi \overline{u})(x)- 
		D^\alpha (\chi \overline{u})(y)\Big)\times \nn\\
		&\qquad \times \Big( D^{\alpha-\beta} \chi(x) D^{\beta} u(x)- D^{\alpha-\beta} \chi(y) D^{\beta} u(y) \Big).
	\end{align}
	Now we estimate the right side of \eqref{eq:HLT-IMS-s>1-a} with the help of the 
	Cauchy-Schwarz inequality. We have 
	\begin{align*}
		&\left| D^{\alpha-\beta} \chi(x) D^{\beta} u(x)- D^{\alpha-\beta} \chi(y) D^{\beta} u(y) \right|^2 \\
		&\quad = \left| D^{\alpha-\beta} \chi(x) (D^{\beta} u(x)-D^\beta u(y)) 
		+ (D^{\alpha-\beta} \chi(x)-D^{\alpha-\beta} \chi(y)) D^{\beta} u(y) \right|^2 \\
		&\quad \leq 2|D^{\alpha-\beta} \chi(x)|^2 |D^{\beta} u(x)-D^\beta u(y)|^2\\
		&\qquad + 2 |D^{\alpha-\beta} \chi(x)-D^{\alpha-\beta} \chi(y)|^2 |D^{\beta} u(y)|^2 \\
		&\quad \leq C \left( |D^{\beta} u(x)-D^\beta u(y)|^2 + |x-y|^2 |D^{\beta} u(y)|^2 \right)
	\end{align*}
	for the second term and 
	\begin{align*}
		&\left| 2\Big(D^\alpha (\chi \overline{u})(x)- D^\alpha (\chi \overline{u})(y)\Big)\Big( D^{\alpha-\beta} \chi(x) 
		D^{\beta} u(x)- D^{\alpha-\beta} \chi(y) D^{\beta} u(y) \Big) \right| \\
		&\quad \leq |x-y|^{2\eps} |D^\alpha (\chi u)(x)- D^\alpha (\chi u)(y)|^2\\
		&\qquad+|x-y|^{-2\eps} \left| D^{\alpha-\beta} \chi(x) D^{\beta} u(x)- D^{\alpha-\beta} \chi(y) D^{\beta} u(y) \right|^2\\
		&\quad \leq |x-y|^{2\eps} |D^\alpha (\chi u)(x)- D^\alpha (\chi u)(y)|^2\\
		&\qquad+  C |x-y|^{-2\eps}
		\left( |D^{\beta} u(x)-D^\beta u(y)|^2 + |x-y|^2 |D^{\beta} u(y)|^2 \right)
	\end{align*}
	for the third term. Here we are choosing $0<\eps<\min\{\sigma, 1-\sigma\}$. 
	When inserting these estimates into \eqref{eq:HLT-IMS-s>1-a} we find
	\begin{align*}
		&\left| |D^\alpha(\chi u)(x)-D^\alpha(\chi u)(y)|^2 - |\chi(x) D^\alpha u (x) - \chi(y) D^\alpha u (y)|^2 \right| \\
		&\quad \le C |x-y|^{2\eps} |D^\alpha (\chi u)(x)- D^\alpha (\chi u)(y)|^2 \\
		&\qquad + C \sum_{\beta<\alpha}(1+|x-y|^{-2\eps}) \left( |D^{\beta} u(x)-D^\beta u(y)|^2 + |x-y|^2 |D^{\beta} u(y)|^2 \right).
	\end{align*}
	Integrating second part of the above inequality against the weight $|x-y|^{-(d+2\sigma)}$ 
	leads to
	\begin{align*}
		&\iint_{\Omega\times \Omega} 
		\frac{\left| |D^\alpha (\chi u)(x)- D^\alpha (\chi u)(y)|^2 - |\chi(x) D^\alpha u (x) - \chi(y) D^\alpha u (y)|^2 \right| }
		{|x-y|^{d+2\sigma}}\,dxdy \\
		&\leq \iint_{\Omega\times \Omega} \frac{|D^\alpha (\chi u)(x)- D^\alpha (\chi u)(y)|^2}{|x-y|^{d+2(\sigma-\eps)}} dxdy \\
		& + C\sum_{\beta<\alpha} \iint_{\Omega\times \Omega} \frac{(1+|x-y|^{-2\eps}) 
		\left( |D^{\beta} u(x)-D^\beta u(y)|^2 + |x-y|^2 |D^{\beta} u(y)|^2 \right)}{|x-y|^{d+2\sigma}}\,dxdy \\
	 	&\leq C\|  D^\alpha (\chi u)\|_{\dot H^{\sigma-\eps}(\Omega)}^2 + C\|u\|_{H^m(\Omega)}^2,
	\end{align*}  
	where we also estimated difference quotients involving $D^\beta u$
	in terms of $D^{\alpha'}u$, $|\alpha'|=m$.
	Combining the above with a similar inequality for $D^\alpha(\eta u)$, we find that
	\begin{align} \label{eq:HLT-IMS-s>1-b}
		&\iint_{\Omega\times \Omega} 
		\frac{\left| |D^\alpha (\chi u)(x)- D^\alpha (\chi u)(y)|^2 - |\chi(x) D^\alpha u (x) - \chi(y) D^\alpha u (y)|^2 \right| }
		{|x-y|^{d+2\sigma}}\,dxdy \nn\\
		& + \iint_{\Omega\times \Omega} 
		\frac{\left| |D^\alpha (\eta u)(x)- D^\alpha (\eta u)(y)|^2 - |\eta (x) D^\alpha u (x) - \eta(y) D^\alpha u (y)|^2 \right| }
		{|x-y|^{d+2\sigma}}\,dxdy \nn\\ 
		& \leq C\|  D^\alpha (\chi u)\|_{\dot H^{\sigma-\eps}(\Omega)}^2 
		+ C\|  D^\alpha (\eta u)\|_{\dot H^{\sigma-\eps}(\Omega)}^2 + C\|u\|_{H^m(\Omega)}^2.
	\end{align}  
	On the other hand, note that as in \eqref{eq:Loss-formula},
	\begin{align*}
		&\Big| |\chi(x) D^\alpha u (x) - \chi(y) D^\alpha u (y)|^2 
		+ |\eta(x) D^\alpha u (x) - \eta(y) D^\alpha u (y)|^2 \\
		&\quad\qquad- |D^\alpha u(x)- D^\alpha u(y)|^2 \Big| \\
		&\quad =  \left| \Big( (\chi(x)-\chi(y))^2 + (\eta(x)-\eta(y))^2 \Big) \Re D^\alpha \overline{u}(x)D^\alpha u(y) \right| \\
		&\quad \leq C |x-y|^2 \Big( |D^\alpha u(x)|^2 + |D^\alpha u(y)|^2 \Big).
	\end{align*}
	Integrating the latter inequality against the weight $|x-y|^{-(d+2\sigma)}$ 
	we get
	\begin{multline}\label{eq:HLT-IMS-s>1-c}
		\bigg| \iint_{\Omega\times \Omega} 
		\frac{|\chi(x) D^\alpha u (x) - \chi(y) D^\alpha u (y)|^2 + |\eta(x) D^\alpha u (x) - \eta(y) D^\alpha u (y)|^2}
		{|x-y|^{d+2\sigma}}\,dxdy\\
		- \iint_{\Omega\times \Omega} \frac{|D^\alpha u (x)-D^\alpha u(y)|^2}{|x-y|^{d+2\sigma}}\,dxdy \bigg| 
		\le C \int_\Omega |D^\alpha u|^2.
	\end{multline}
	From \eqref{eq:HLT-IMS-s>1-b}-\eqref{eq:HLT-IMS-s>1-c} and the 
	triangle inequality, it follows that 
	\begin{align*}
		&\bigg| \iint_{\Omega\times \Omega} 
		\frac{|D^\alpha (\chi u)(x)- D^\alpha (\chi u)(y)|^2 + |D^\alpha (\eta u)(x)- D^\alpha (\eta u)(y)|^2 }{|x-y|^{d+2\sigma}}\,dxdy\\
		& \qquad - \iint_{\Omega\times \Omega} \frac{|D^\alpha u (x)-D^\alpha u(y)|^2}{|x-y|^{d+2\sigma}} dxdy \bigg| \\
		&\quad \leq C\|  D^\alpha (\chi u)\|_{\dot H^{\sigma-\eps}(\Omega)}^2 
		+ C\|  D^\alpha (\eta u)\|_{\dot H^{\sigma-\eps}(\Omega)}^2 + C\|u\|_{H^m(\Omega)}^2
	\end{align*}  
	for all $|\alpha|=m$. By taking the sum over all $\alpha$'s with $|\alpha|=m$, 
	we get
	\begin{multline*}
		\left| \|\chi u\|_{\dot H^s(\Omega)}^2 + \|\eta u\|_{\dot H^s(\Omega)}^2 - \|u\|_{\dot H^s(\Omega)}^2 \right| \\
		\le C\Big( \|  \chi u\|_{\dot H^{s-\eps}(\Omega)}^2 + \|  \eta u\|_{\dot H^{s-\eps}(\Omega)}^2 + \|u\|_{H^m(\Omega)}^2 \Big).
	\end{multline*}
	Combining this with the estimate  
	$$
		\|u\|_{H^m(\Omega)}^2\le C (\|\chi u\|_{H^m(\Omega)}^2 + \|\eta u\|_{H^m(\Omega)}^2),
	$$
	which follows from the integer case in Step 1, we can conclude that
	\begin{align*}
		\left| \|\chi u\|_{\dot H^s(\Omega)}^2 + \|\eta u\|_{\dot H^s(\Omega)}^2 - \|u\|_{\dot H^s(\Omega)}^2 \right| 
		\leq C\Big( \|  \chi u\|_{\dot H^{s-\eps}(\Omega)}^2 + \|\eta u\|_{\dot H^{s-\eps}(\Omega)}^2 \Big).
	\end{align*}
	This is the desired inequality. 
\end{proof}

\subsection{Proof of the Hardy-Lieb-Thirring inequality}

\begin{proof}[Proof of Theorem~\ref{thm:HLT_frac}]
	By a standard approximation argument we can assume that $\rho_\Psi$ is 
	supported in a finite cube $Q_0 \subset \R^d$ which centers at $0$. 
	Let an arbitrary $0<\Lambda\le N$. 
	By Lemma~\ref{lem:covering} with $f=\rho_\Psi$, $k=3$ and $\alpha=2s/d$, there exists a division of $Q_0$
	into disjoint sub-cubes $Q$'s such that $\int_Q \rho_\Psi \le \Lambda$
	and
	\bq \label{eq:bound-f*f-f-3d}
		\sum_{Q} \frac{1}{|Q|^{\alpha}} \left[ \left(\int_{Q} f \right)^2   - \frac{\Lambda}{b} \int_{Q} f \right] \ge 0,
	\eq
	with
	$$
		b := \frac{3^d}{2}\left(1 + \sqrt{1 + \frac{1-3^{-d}}{3^{d\alpha}-1} } \right).
	$$
	Moreover, for every sub-cube $Q$ we have either that $Q$ centers at $0$ 
	or that $\inf_{x\in Q}|x|\ge |Q|^{1/d}/2$. 
		
	Now we claim that there exists a constant $C_1>0$ depending only on $d\ge 1$ 
	and $s>0$ such that for every sub-cube $Q$ and for every function $u\in H^s(Q)$ 
	we have the uncertainty relation
	\begin{align} \label{eq:HLT-uncertainty-proof}
		\|u\|_{\dot H^s(Q)} - \mathcal{C}_{d,s}\int_Q \frac{|u(x)|^2}{|x|^{2s}}\,dx 
		\ge \frac{1}{C_1}\frac{\int_Q |u|^{2(1+2s/d)}}{\Big(\int_Q |u|^2\Big)^{2s/d}} - \frac{C_1}{|Q|^{2s/d}} \int_Q |u|^2.
	\end{align}
	In fact, if $Q$ centers at $0$, then \eqref{eq:HLT-uncertainty-proof} is 
	covered by Lemma \ref{lem:HLT-uncertainty}. 
	On the other hand, if $0\notin Q$, then using $|x|\ge |Q|^{1/d}/2$ we have
	$$
		\int_Q \frac{|u|^2}{|x|^{2s}}\,dx \le \frac{2^{2s}}{|Q|^{2s/d}} \int_Q |u(x)|^2 dx
	$$
	and \eqref{eq:HLT-uncertainty-proof} is covered by Lemma~\ref{lem:local_uncert}. 
	Using \eqref{eq:HLT-uncertainty-proof} 
	and arguing in exactly the same way as in the 
	proof of Lemma~\ref{lem:many-body-local-uncertainty}, 
	we obtain the many-body estimate
	\begin{multline}\label{eq:many-body-local-uncertainty-HLT}
		\left\langle \Psi, \sum_{i=1}^N \Big((-\Delta_{i})^s - \mathcal{C}_{d,s}|x|^{-2s} \Big) \Psi \right\rangle
		\ge \sum_Q \left[ \frac{1}{C_1}\frac{\int_Q \rho_\Psi^{1+2s/d}}{\Big(\int_Q \rho_\Psi \Big)^{2s/d}} 
		- \frac{C_1}{|Q|^{2s/d}} \int_Q \rho_\Psi \right]\\
		\ge\frac{1}{C_1 \Lambda^{2s/d}} \int_{\R^d}  \rho_\Psi^{1+2s/d} - \sum_Q \frac{C_1}{|Q|^{2s/d}} \int_Q \rho_\Psi .
	\end{multline}
	Here in the last inequality of \eqref{eq:many-body-local-uncertainty-HLT} we 
	have used the bound $\int_Q \rho_\Psi \le \Lambda$ for all $Q$. 
	Combining \eqref{eq:many-body-local-uncertainty-HLT}, 
	Lemma~\ref{lem:local-ex} and \eqref{eq:bound-f*f-f-3d}, we find that
	\begin{align}\label{eq:HLT-final}
		&\left\langle \Psi, \left( \sum_{i=1}^N \Big((-\Delta_{i})^s - \mathcal{C}_{d,s}|x|^{-2s} \Big)
		+\sum_{1\le i<j \le N} \frac{1}{|x_i-x_j|^{2s}} \right) \Psi \right\rangle\nn\\
		&\ge \frac{1}{C_1 \Lambda^{2s/d}} \int_{\R^d}  \rho_\Psi^{1+2s/d}  
		+ \sum_Q \frac{1}{2d^s|Q|^{2s}} \left( \Big(\int_Q \rho_\Psi \Big)^2 - (2d^sC_1+1) \int_Q \rho_\Psi \right)\nn\\
		&\ge \frac{1}{C_1 \Lambda^{2s/d}} \int_{\R^d}  \rho_\Psi^{1+2s/d}  
		+ \left( \frac{\Lambda}{b}- 2d^sC_1 - 1 \right) \sum_Q \frac{1}{2d^s|Q|^{2s}} \int_Q \rho_\Psi
	\end{align}
	for all $0<\Lambda \le N$. 
	
	On the other hand, using the interpolation inequality 
	\eqref{eq:Hardy-improved-GN} with
	$$ 
		q=\frac{2d}{d-2t}=2 \left( 1+\frac{2s}{d}\right), 
		\quad \text{that is} \ \ t=\frac{ds}{d+2s},
	$$
	and the same argument of the proof of Lemma \ref{lem:many-body-local-uncertainty}, 
	we obtain the following strengthened version of \eqref{eq:LT-proof-N-small}:
	\begin{align}\label{eq:HLT-proof-N-small}
		\left\langle \Psi, \sum_{i=1}^N \Big( (-\Delta_{i})^s -\mathcal{C}_{d,s}|x|^{2s}\Big) \Psi \right\rangle 
		\ge C N^{-2s/d} \int_{\R^d} \rho_\Psi^{1+2s/d},
	\end{align}
	for a constant $C>0$ depending only on $d$ and $s$. 
	
	Finally, using \eqref{eq:HLT-final} with $\Lambda=(2d^sC_1 + 1)b=:\Lambda_0$ 
	if $N>\Lambda_0$, and using \eqref{eq:HLT-proof-N-small} if $N\le \Lambda_0$, 
	we find the desired inequality.
\end{proof}

\begin{remark}
	Also in this case it is possible to add a coupling parameter $\lambda>0$
	as in \eqref{eq:LT_frac_coupling},
	and a straightforward adaptation of \eqref{eq:HLT-final} yields
	for the corresponding constant
	$C(\lambda) \sim \min\{1,\lambda^{2s/d}\}$.
\end{remark}

\section{Interpolation inequalities}\label{sec:LT-interpolation}
\subsection{Equivalence for the Lieb-Thirring inequality} 
In this subsection, we provide a proof of Theorem~\ref{thm:LT_equiv},
i.e. the equivalence of the Lieb-Thirring inequality \eqref{eq:LT_frac}
and the one-body interpolation inequality \eqref{eq:LT-inter-1}. 
The implication of \eqref{eq:LT-inter-1} from \eqref{eq:LT_frac} was already 
explained in Section~\ref{ssec:interpolation} and it holds for all $0<s<d/2$. 
In the following, we show that the interpolation inequality 
\eqref{eq:LT-inter-1} implies the Lieb-Thirring inequality \eqref{eq:LT_frac} 
when $0<s<d/2$ and $s\le 1$. 
 
We will use the Hoffmann-Ostenhof and Lieb-Oxford inequalities, 
which reduce the kinetic and interaction energies of a many-body state to those of its density.

\begin{lemma}[Hoffmann-Ostenhof inequality]\label{lem:HO}
	For every $0<s\le 1$ and every normalized function $\Psi\in L^2((\R^d)^N)$, one has
	\bq \label{eq:Hof}
		\left\langle \Psi, \sum_{i=1}^N (-\Delta_{i})^s\Psi \right\rangle 
		\ge \langle  \sqrt{\rho_\Psi}, (-\Delta)^s \sqrt{\rho_\Psi} \rangle.
	\eq
\end{lemma}
The non-relativistic case $s=1$ of \eqref{eq:Hof} was first discovered by 
M. \& T. Hoffman-Ostenhof \cite{Hof-77}. 
In fact, \eqref{eq:Hof} is equivalent to the one-body inequality 
$\langle u, (-\Delta)^s u\rangle \ge \langle |u|, (-\Delta)^s |u|\rangle $ 
(cf. the diamagnetic inequality \eqref{eq:diam_ineq}) 
and it is false when $s>1$.
See e.g. \cite[Lemma 8.4]{LieSei-10} for a proof of \eqref{eq:Hof} 
and further discussions. 

\begin{lemma}[Lieb-Oxford inequality for homogeneous potentials]\label{lem:LO}
	For every $0<\gamma<d$ and for every normalized function $\Psi\in L^2((\R^d)^N)$, one has
	\bq \label{eq:LO}
		\left\langle \Psi, \sum_{1\le i<j \le N} \frac{1}{|x_i-x_j|^{\gamma}}  \Psi \right\rangle 
		\ge \frac{1}{2} \iint \frac{\rho_\Psi(x)\rho_\Psi(y)}{|x-y|^{\gamma}}\,dxdy -C_{\rm LO} \int\rho_\Psi^{1+\gamma/d}
	\eq
	for a constant $C_{\rm LO}>0$ depending only on $d$ and $\gamma$.
\end{lemma}

The case $\gamma=1$ and $d=3$ of \eqref{eq:LO} was first studied in 
\cite{Lieb-79,LieOxf-80}. 
The case $\gamma=1$ and $d=2$ was proved in \cite[Lemma 5.3]{LieSolYng-95}. 
A proof of Lemma \ref{lem:LO} following the strategy in \cite{LieSolYng-95} 
is provided in Appendix \ref{apd:LO}.  

We are now in a position to complete the proof of equivalence.
\begin{proof}[Proof of Theorem~\ref{thm:LT_equiv}] 
	We prove that \eqref{eq:LT-inter-1} implies \eqref{eq:LT_frac} 
	when $0<s<d/2$ and $s\le 1$. 
	By the Hoffmann-Ostenhof inequality \eqref{eq:Hof} and the 
	Lieb-Oxford inequality \eqref{eq:LO}, one has
	\begin{multline*} 
		\left\langle \Psi, \left( \sum_{i=1}^N (-\Delta_{i})^{s} + \sum_{1\le i<j \le N} \frac{1}{|x_i-x_j|^{2s}} \right) \Psi \right\rangle \\
		\ge \langle  \sqrt{\rho_\Psi}, (-\Delta)^s \sqrt{\rho_\Psi} \rangle 
		+ \frac{\eps}{2} \iint_{\R^d \times \R^d} \frac{\rho_\Psi(x)\rho_\Psi(y)}{|x-y|^{2s}}\,dxdy 
		-  \eps  C_{\rm LO} \int_{\R^d} \rho_\Psi^{1+2s/d}
	\end{multline*}
	for every $\eps\in (0,1]$. On the other hand, by using Young's inequality and the interpolation 
	inequality \eqref{eq:LT-inter-1} with $u = \sqrt{\rho_\Psi}$, we obtain
	\begin{align*} 
		&\left(1-\frac{2s}{d} \right)\langle  \sqrt{\rho_\Psi}, (-\Delta)^s \sqrt{\rho_\Psi} \rangle 
		+ \eps \frac{2s}{d} \iint_{\R^d \times \R^d} \frac{\rho_\Psi(x)\rho_\Psi(y)}{|x-y|^{2s}}\,dxdy \\
		&\quad \ge \eps^{2s/d}\langle  \sqrt{\rho_\Psi}, (-\Delta)^s \sqrt{\rho_\Psi} \rangle^{1-2s/d}
		 \left( \iint_{\R^d \times \R^d} \frac{\rho_\Psi(x)\rho_\Psi(y)}{|x-y|^{2s}}\,dxdy \right)^{2s/d} \\
		&\quad \ge C \eps^{2s/d} \int \rho_\Psi^{1+2s/d}
	\end{align*}  
	for a constant $C >0$ depending only on $d$ and $s$. Thus 
	\begin{align*} 	
		\left\langle \Psi, \left( \sum_{i=1}^N (-\Delta_{i})^s +  \sum_{1\le i<j \le N} \frac{1}{|x_i-x_j|^{2s}} \right) \Psi \right\rangle 
		\ge  \Big( C \eps^{2s/d}  -  C_{\rm LO} \eps \Big) \int \rho_\Psi^{1+2s/d}
	\end{align*}
	for all $\eps\in (0,1]$. As $2s/d<1$, we can choose $\eps>0$ small enough such that 
	$$
		C \eps^{2s/d}  -  C_{\rm LO} \eps>0.
	$$
	Then the Lieb-Thirring inequality \eqref{eq:LT_frac} follows.
\end{proof}

\subsection{Isoperimetric inequality with non-local term}

In the following we show how to use our local approach to Lieb-Thirring
inequalities to prove the one-body interpolation inequality in
Theorem~\ref{thm:iso}.

\begin{proof}[Proof of Theorem~\ref{thm:iso}]
	By a standard approximation argument, we can assume that $u$ is supported in 
	a finite cube $Q_0 \subset \R^d$. 
	Let $f(x) := |u(x)|^{2s}$.
	For an arbitrary $0<\Lambda \le \int_{\R^d} f$, 
	we divide $Q_0$ into disjoint 
	sub-cubes $Q$'s by applying Covering Lemma \ref{lem:covering} with $k=2$ and $\alpha=2s/d$. 
	Thus we have $\int_Q f \le \Lambda$ for all cubes $Q$'s and
	\begin{align} \label{eq:iso-covering-gen1}
		\sum_{Q} \frac{1}{|Q|^{\alpha}} \left[ \left(\int_{Q} f \right)^2 
			- \frac{\Lambda}{a} \int_{Q} f \right] \ge 0, 
		\quad a:= \frac{2^d}{2}\left(1 + \sqrt{1 + \frac{1-2^{-d}}{2^{2s}-1} } \right).
	\end{align}
	
	Similarly to the proof of Lemma~\ref{lem:local-ex}, 
	by ignoring the interaction energy between different cubes and using 
	$|x-y|\le \sqrt{d}|Q|^{1/d}$ for $x,y\in Q$, we have
	\begin{align} \label{eq:interaction-lower-gen1}
		&\iint_{\R^d\times \R^d} \frac{f(x)f(y)}{|x-y|^{2s}} 
		\ge \sum_{Q} \iint_{Q\times Q} \frac{f(x)f(y)}{|x-y|^{2s}} 
		\ge \sum_{Q} \frac{1}{d^s|Q|^{2s/d}} \left( \int_{Q} f \right)^2 .
	\end{align}
	On the other hand, by the Sobolev inequality 
	(recall that $1 \le 2s < d$)
	\begin{align}  \label{eq:iso-Sobolev-Q-gen1}
		\|u\|_{W^{1,2s}(Q)} \ge C\|u\|_{L^q(Q)},
		\quad q = \frac{2sd}{d-2s} > 2s,
	\end{align} 
	we have
	$$
		\|u\|_{W^{1,2s}(Q)}^{2s} 
		\ge C\|f\|_{L^{\frac{d}{d-2s}}(Q)}
		\ge C\frac{\int_Q f^{1+2s/d}}{\left(\int_Q f\right)^{2s/d}}.
	$$
	Hence,
	\begin{multline*}
		\int_{\R^d} |\nabla u|^{2s} 
			+ \sum_Q \frac{1}{|Q|^{2s/d}} \int_Q |u|^{2s}
		= \sum_Q \left( \int_Q |\nabla u|^{2s} 
			+ \frac{1}{|Q|^{2s/d}} \int_Q |u|^{2s} \right) \\
		\ge \sum_Q 2^{1-2s} \|u\|_{W^{1,2s}(Q)}^{2s}
		\ge C\sum_Q \frac{\int_Q |u|^{2s(1+2s/d)}}{\left(\int_Q f\right)^{2s/d}},
	\end{multline*}
	and, combining with \eqref{eq:interaction-lower-gen1}
	and \eqref{eq:iso-covering-gen1},
	\begin{align*}
		&\int_{\R^d} |\nabla u|^{2s} dx
		+ \iint_{\R^d\times \R^d} \frac{|u(x)|^{2s} |u(y)|^{2s}}{|x-y|^{2s}}\,dxdy\\
		&\quad \ge \frac{C_1}{\Lambda^{2s/d}} \int_{\R^d} |u|^{2s(1+2s/d)} 
		+ \sum_Q \frac{1}{|Q|^{2s/d}} \left(\frac{1}{d^s}\left(\int_Q f\right)^2 - \int_Q f \right) \\
		&\quad \ge \frac{C_1}{\Lambda^{2s/d}} \int_{\R^d} |u|^{2s(1+2s/d)} 
		+ \left( \frac{\Lambda}{d^s a} - 1 \right) \sum_Q \frac{1}{|Q|^{2s/d}} \int_Q f.
	\end{align*}
	Thus, if $\int_{\R^d} f \ge d^s a$, then we can simply choose 
	$\Lambda=d^s a$ and conclude that
	\begin{align*}
		& \int_{\R^d} |\nabla u|^{2s} 
			+ \iint_{\R^d\times \R^d} \frac{|u(x)|^{2s} |u(y)|^{2s}}{|x-y|^{2s}}\,dxdy 
		\ge \frac{C_1}{(d^s a)^{2s/d}} \int_{\R^d} |u|^{2s(1+2s/d)}.
	\end{align*}
	On the other hand, if $\int_{\R^d} f \le d^s a$, then 
	using Sobolev's inequality
	\begin{align} \label{eq:iso-Sobolev-Rd}
		\|\nabla u\|_{L^{2s}(\R^d)} \ge C_2\|u\|_{L^{2sd/(d-2s)}(\R^d)},
		\quad \forall u\in W^{1,2s}(\R^d)
	\end{align}
	and H\"older's inequality we have
	\begin{align*}
		& \int_{\R^d} |\nabla u|^{2s} \ge C_2 \|f\|_{L^{d/(d-2s)}(\R^d)} 
		\ge C_2 \frac{\int_{\R^d} f^{1+2s/d}}{\Big(\int_{\R^d} f\Big)^{2s/d}} 
		\ge \frac{C_2}{(d^s a)^{2s/d}} \int_{\R^d} |u|^{2s(1+2s/d)}.
	\end{align*}
	In summary, it always holds that
	\begin{align} \label{iso-LT-type}
		& \int_{\R^d} |\nabla u|^{2s} dx
			+ \iint_{\R^d\times \R^d} \frac{|u(x)|^{2s} |u(y)|^{2s}}{|x-y|^{2s}}\,dxdy 
		\ge \frac{\min\{C_1,C_2\}}{(d^s a)^{2s/d}} \int_{\R^d} |u|^{2s(1+2s/d)}.
	\end{align}
	
	By proceeding as for the Lieb-Thirring inequality in Section~\ref{ssec:interpolation},
	that is rescaling $u \mapsto \mu u$ and optimizing over $\mu>0$,
	we obtain the interpolation inequality \eqref{eq:LT-inter-3}.
\end{proof}

\appendix

\section{Lieb-Oxford inequality for homogeneous potentials} \label{apd:LO}
In this appendix we prove Lemma \ref{lem:LO}. Note that the argument in the original papers \cite{Lieb-79,LieOxf-80} uses Newton's theorem and hence only works with 
the standard Coulomb interaction. The following proof is based on the strategy of Lieb, Solovej and Yngvason \cite[Lemma 5.3]{LieSolYng-95}. 

\begin{proof} [Proof of Lemma \ref{lem:LO}]
	We start with the Fefferman-de la Llave representation 
	 \[
	 	\frac{1}{{|x - y|^\gamma}} 
		= c_{d,\gamma}\int_0^{\infty}\int_{\R^d}\1_{B_R}(x - u)\1_{B_R}(y - u)du\,\frac{dR}{R^{d+\gamma+1}},
	 \]
	where $B_R=\overline{B(0,R)}$ is the closed ball in $\R^d$ and $c_{d,\gamma}$ is a constant 
	depending only on $d$ and $\gamma$ (see \cite{FefLla-86} for Coulomb potential, 
	\cite[Theorem 9.8]{LieLos-01} for homogeneous potentials and \cite[Theorem 1]{HaiSei-02} for more general cases). 
	Consequently,
	\bq\label{eq:LO-eq1}
		\iint_{\R^d \times \R^d} \frac{\rho_\Psi (x)\rho_\Psi (y)}{|x - y|^\gamma}\,dxdy
		= \int_0^{\infty} \int_{\R^d} f_R(u)^2du\,\frac{dR}{R^{d+\gamma+1}},
	\eq
	where
	$$
		f_R :=  \rho_\Psi*\1_{B_R}
	$$
	and  
	\begin{align} \label{eq:LO-eq2}
		\left\langle \Psi, {  \sum\limits_{1 \leqslant i < j \leqslant n} {\frac{1}{{|x_i  - x_j |^\gamma}}} } \Psi \right\rangle 
		=  c_{d,\gamma} \int_0^\infty \int_{\R^d}g_R(u)du\,\frac{dR}{R^{d+\gamma+1}}
	\end{align}
	where
	$$
		g_R(u): = \left \langle \Psi, \sum\limits_{1 \leqslant i < j \leqslant N} \1_{B_R} (x_i  - u)\1_{B_R} (x_j  - u) \Psi \right \rangle .
	$$
	Using the Cauchy-Schwarz inequality we find that 
	\begin{align*}
		g_R(u) &= \frac{1}{2}\left \langle \Psi, \Big(\sum\limits_{i=1}^N \1_{B_R} (x_i  - u) \Big)^2 \Psi \right \rangle 
		-\frac{1}{2} \left \langle \Psi, \sum\limits_{i=1}^N \1_{B_R} (x_i  - u) \Psi \right \rangle\\
	    	&\ge \frac{1}{2} \left \langle \Psi, \sum\limits_{i=1}^N \1_{B_R} (x_i  - u) \Psi \right \rangle ^2 
		- \frac{1}{2} \left \langle \Psi, \sum\limits_{i=1}^N \1_{B_R} (x_i  - u) \Psi \right \rangle \\
		& = \frac{1} {2}f^2_R(u) - \frac{1} {2}f_R(u).
	\end{align*}
	Combining with the obvious inequality $g_R(u)\ge 0$ we get
	 $$
		g_R(u) \ge \frac{1}{2}f^2_R(u) - \frac{1}{2}\min \{ f_R(u),f^2_R(u)\}.
	 $$
	 Inserting the latter inequality into \eqref{eq:LO-eq2} and using \eqref{eq:LO-eq1}, we conclude that
	\begin{align}\label{eq:LO-eq3}
		\left\langle \Psi, \sum_{1 \leqslant i < j \leqslant n} \frac{1}{{|x_i  - x_j |^\gamma}} \Psi \right\rangle 
		&\ge \frac{1}{2}\iint_{\R^d \times \R^d} {\frac{{\rho_\Psi (x)\rho_\Psi (y)}}{{|x - y|^2}}dxdy} \\
		&\quad -\frac{c_{d,\gamma}}{2} \int_0^\infty  {\int_{\R^d } \min \{ f_R(u),f^2_R(u)\} 
		du\frac{{dR}}  {{R^{d+\gamma+1} }}} \nn.
	\end{align}
	To estimate the second term of the right side, we introduce the Hardy-Littlewood maximal function of $\rho_\Psi$:
	 $$
	 	\rho ^* (u) := \sup_{R > 0} \frac{1}{{|B(0,R)|}}\int_{|x-u|\le R } {\rho_\Psi (x)dx}  = |B_1|^{-1} \mathop {\sup }
		\limits_{R > 0} \frac{{ f_R(u)}}{{R^d }}.
	 $$
	Using $f_R(u)\le |B_1| R^{d} \rho^*(u)$, we find that 
	\begin{align*}
	  	&\int_{0}^\infty  \min\{ f_R^2 (u),f_R (u)\} \frac{dR}{R^{d+\gamma+1}} 
		\le \int_0^{R_*} f_R^2(u)\,\frac{dR}{R^{d+\gamma+1}}  
		\,+\, \int_{R_* }^\infty f_R (u)\,\frac{dR}{R^{d+\gamma+1}}\\
	    	&\quad \le \int_0^{R_* } \left( |B_1| R^d \rho^*(u)\right)^2 \frac{dR}{R^{d+\gamma+1}}
		+ \int_{R_* }^\infty  |B_1|R^{d} \rho^*(u)\,\frac{dR}{R^{d+\gamma+1}}\\
	     	&\quad =\frac{|B_1|^2}{d-\gamma} R_*^{d-\gamma} (\rho ^* (u))^2  
		+\frac{|B_1|}{\gamma} R_*^{-\gamma} \rho ^* (u)
	\end{align*}
	for all $u\in \R^d$ and for all $R_*>0$. Choosing $R_* =( |B_1| \rho^*(u) )^{-1/d}$, we get 
	\begin{align*}
		\int_{0}^\infty  \min \{ f_R^2 (u),f_R (u)\}\,\frac{dR}{R^{d+\gamma+1}}
		\le \frac{d}{\gamma(d-\gamma)} |B_1|^{1+\gamma/d} (\rho^*(u))^{1+\gamma/d}
	\end{align*}
	for all $u\in \R^d$. Finally, by the maximal inequality (see, e.g. \cite[p.58]{SteWei-71})
	$$
		\int_{\R^d} (\rho^*(u))^{1+\gamma/d} du 
		\le M_{d,\gamma} \int_{\R^d} \rho_\Psi(u)^{1+\gamma/d} du,
	$$
	where $M_{d,\gamma}$ is a constant depending only on $d$ and $\gamma$, we conclude from \eqref{eq:LO-eq3} that
	\begin{multline*}
		\left\langle \Psi, {  \sum\limits_{1 \leqslant i < j \leqslant n} {\frac{1}{{|x_i  - x_j |^\gamma}}} } \Psi \right\rangle\\
		\ge \frac{1}{2}\iint_{\R^d \times \R^d} {\frac{{\rho_\Psi (x)\rho_\Psi (y)}}{{|x - y|^\gamma}}} dxdy
		- \frac{dc_{d,\gamma} 
		M_{d,\gamma} }{2\gamma(d-\gamma)} |B_1|^{1+\gamma/d} 
		\int_{\R^d} \rho_\Psi ^{1+\gamma/d}.
	\end{multline*}
	This is the desired inequality.
\end{proof}

\bibliographystyle{siam}

\end{document}